\documentclass[12pt]{article}
\usepackage[top=1in,bottom=1in,left=1in,right=1in]{geometry}
\usepackage[utf8]{inputenc}
\usepackage{amsmath}
\usepackage{amssymb}
\usepackage{amsthm}
\usepackage{xcolor}
\usepackage{graphicx}
\usepackage{float}
\usepackage{enumerate}
\usepackage{natbib}
\usepackage{setspace}
\usepackage{mathrsfs}
\usepackage{comment}
\usepackage{apptools}
 \usepackage{booktabs}
\usepackage{siunitx}
\usepackage{xr}
\newcommand{\RN}[1]{  \textup{\uppercase\expandafter{\romannumeral\#1}}}

\newtheorem{proposition}{Proposition}

\newtheorem{lemma}{Lemma}

\newtheorem{assumption}{Assumption}

\theoremstyle{definition}
\newtheorem{example}{Example}
\theoremstyle{remark}

\newtheorem{axiom}{Axiom}
\newtheorem{definition}{Definition}
\usepackage{setspace}

\makeatother

\AtAppendix{
\numberwithin{equation}{section}
\numberwithin{definition}{section}
\numberwithin{theorem}{section}
\numberwithin{proposition}{section}
\numberwithin{lemma}{section}}

\DeclareMathOperator*{\argmin}{arg\,min}

\makeatletter
\newcommand*{\addFileDependency}[1]{
\typeout{(#1)}
\@addtofilelist{#1}
\IfFileExists{#1}{}{\typeout{No file #1.}}
}\makeatother

\newcommand*{\myexternaldocument}[1]{%
\externaldocument{#1}%
\addFileDependency{#1.tex}%
\addFileDependency{#1.aux}%
}
\myexternaldocument{onlineappendix}

\allowdisplaybreaks

\title{How Flexible is that Functional Form? \\ Quantifying the Restrictiveness of Theories\thanks{We thank Nikhil Agarwal, Victor Aguiar, Abhijit Banerjee, Tilman B\"{o}rgers,  Vincent Crawford, Glenn Ellison, Benjamin Enke, Ben Golub, Taisuke Imai, Shaowei Ke, David Laibson, Rosa Matzkin, John Quah, Kareen Rozen, Jesse Shapiro, Ludvig Sinander, Charles Sprenger, Dmitry Taubinsky, and Emanuel Vespa for helpful comments, and NSF grants SES 185162 and 1951056 for financial support. We thank Kyohei Okumura for excellent research assistance.}}

\author{Drew Fudenberg\thanks{Department of Economics, MIT; drewf@mit.edu}  \quad \quad Wayne Gao \thanks{Department of Economics, University of Pennsylvania; waynegao@upenn.edu}\quad \quad Annie Liang\thanks{Department of Economics, Northwestern University; annie.liang@northwestern.edu. }}

\begin{document}

\global\long\def\a{\alpha}%
\global\long\def\b{\beta}%
\global\long\def\g{\gamma}%
\global\long\def\d{\delta}%
\global\long\def\e{\varepsilon}%
\global\long\def\l{\lambda}%
\global\long\def\t{\theta}%
\global\long\def\o{\omega}%
\global\long\def\s{\sigma}
\global\long\def\G{\Gamma}%
\global\long\def\D{\Delta}%
\global\long\def\L{\Lambda}%
\global\long\def\T{\Theta}%
\global\long\def\O{\Omega}%
\global\long\def\R{\mathbb{R}}%
\global\long\def\N{\mathbb{N}}%
\global\long\def\P{\mathbb{P}}%
\global\long\def\Q{\mathbb{Q}}%
\global\long\def\X{\mathbb{X}}%
\global\long\def\U{\mathbb{U}}%
\global\long\def\E{\mathbb{E}}%
\global\long\def\N{\mathbb{N}}%
\global\long\def\Q{\mathbb{Q}}%
\global\long\def\cX{{\cal X}}%
\global\long\def\cY{{\cal Y}}%
\global\long\def\cA{{\cal A}}%
\global\long\def\cB{\mathscr{B}}
\global\long\def\cN{{\cal N}}%
\global\long\def\cM{{\cal M}}%
\global\long\def\cG{{\cal G}}%
\global\long\def\cF{{\cal F}}%
\global\long\def\cFol{\ol{\cal F}}%
\global\long\def\cP{{\cal P}}%
\global\long\def\es{\emptyset}%
\global\long\def\ind{\mathbf{1}}%
\global\long\def\any{\forall}%
\global\long\def\ex{\exists}%
\global\long\def\p{\partial}%
\global\long\def\cd{\cdot}%
\global\long\def\Dif{\nabla}%
\global\long\def\imp{\Rightarrow}%
\global\long\def\iff{\Leftrightarrow}
\global\long\def\up{\uparrow}%
\global\long\def\down{\downarrow}%
\global\long\def\arrow{\rightarrow}%
\global\long\def\rlarrow{\leftrightarrow}%
\global\long\def\lrarrow{\leftrightarrow}
\global\long\def\abs#1{\left|#1\right|}%
\global\long\def\norm#1{\left\Vert #1\right\Vert }%
\global\long\def\rest#1{\left.#1\right|}%
\global\long\def\inprod#1{\left\langle #1\right\rangle }%
\global\long\def\ol#1{\overline{#1}}%
\global\long\def\ul#1{\underline{#1}}%
\global\long\def\td#1{\tilde{#1}}
\global\long\def\upto{\nearrow}%
\global\long\def\downto{\searrow}%
\global\long\def\pto{\overset{p}{\longrightarrow}}%
\global\long\def\dto{\overset{d}{\longrightarrow}}%
\global\long\def\asto{\overset{a.s.}{\longrightarrow}}%

\sloppy

\maketitle
\thispagestyle{empty}
\vspace{-2em}

\begin{abstract}
We propose a  \emph{restrictiveness}  measure for economic models based on how well they fit synthetic data from a pre-defined class. This measure, together with a measure for how well the model fits real data, outlines a Pareto frontier, where models that rule out more regularities, yet capture the regularities that are present in real data, are preferred. To illustrate our approach, we evaluate the restrictiveness of popular models in two laboratory settings---certainty equivalents and initial play---and in one field setting---takeup of microfinance in Indian villages. The restrictiveness measure reveals new insights about each of the models, including that some economic models with only a few parameters are very flexible.
\end{abstract}
\newpage

\setcounter{page}{1}

\section{Introduction}
If a parametric model fits the data well, is it because the model captures structure specific to the observed data, or because the model is so flexible that it would fit almost all conceivable data?  This paper provides a quantitative measure of restrictiveness that can distinguish between these  two explanations, and is easy to compute in  a variety of applications.

Our approach for evaluating the restrictiveness of a model is to generate synthetic data sets, and  evaluate how well the model fits this synthetic data.  Some models have known properties, for example  Cumulative Prospect Theory requires that certainty equivalents for lotteries respect first-order stochastic dominance. For these models, the relevant question may not be whether the model is restrictive at all, but instead how much content it has beyond these known constraints. We define the \emph{eligible} data to be those data sets that satisfy specified background constraints, and measure a model's restrictiveness by its (normalized) average error across the eligible data.

We complement the evaluation of restrictiveness, which is based solely on synthetic data, with an evaluation of the model's performance on actual data, using the completeness measure proposed in \citet{FKLM}. Restrictiveness and completeness provide complementary perspectives, and define a Pareto frontier where models that rule out more regularities, yet capture the regularities that are present in real data, are preferred.\footnote{These are not the only considerations that matter for evaluating models, and we do not speak to other important concerns such as parameter estimation and causal inference. Nevertheless, these two measures may be relevant to those problems as well: If a model can fit almost any data set, then its good fit to a specific real data set does not necessarily mean that the model is the ``right" model.}

 Section \ref{sec:axiomatic} provides axioms for our restrictiveness measure to clarify its theoretical properties. The main axioms require that the measure is homogeneous in the unit scale used to quantify model error, and that the measure has a linearity property as the background constraints are varied. An additional ``symmetry" axiom requires that the model's ability to approximate different synthetic data sets has the same effect on the restrictiveness measure. Dropping this axiom returns a broader class of restrictiveness measures, where instead of averaging across synthetic data sets, the data sets are weighted by an analyst's prior. We  develop estimators for both the restrictiveness and completeness measures in Section \ref{Est}, and establish their asymptotic properties so that users can compute confidence intervals.

A key feature of our restrictiveness measure is that is  computable without the guidance of theoretical results about the model's implications or empirical content. This differentiates restrictiveness from measures such as the model's VC dimension, or its hit-rate and accuracy-rate as defined in \citet{Selten}.\footnote{There are representation theorems for many non-parametric theories of individual choice, and some analytic results for the sets of equilibria in games, but we are unaware of representation theorems for most functional forms that are commonly used in applied work.} (Section \ref{sec:related} reviews the related literature and relates it to our work.) The measure's tractability makes it  easy to apply to a variety of contexts, as  we demonstrate by applying it to models from three economic domains: (1) predicting certainty equivalents for binary lotteries (where we evaluate \emph{Cumulative Prospect Theory} and \emph{Disappointment Aversion}); (2) predicting initial play in matrix games (where we evaluate the \emph{Poisson Cognitive Hierarchy Model (PCHM)}, \emph{Logit PCHM}, and \emph{Logit Level-1}); and (3) predicting takeup of microfinance in Indian villages (where we evaluate linear regression models based on economically-motivated regressors, and a structural model of diffusion).\footnote{In addition to these applications, \citet{Schwaninger}  uses  our restrictiveness measure to evaluate models of bargaining with inequity aversion, \citet{EllisKarizOzbay} uses it to evaluate models of consumer demand from budget sets, and \citet{ba2023over} uses it to evaluate models of reaction to information.} The first two settings use data from the lab, our third application uses field data. 
In each of these domains, these measures reveal new insights about the models we examine, which we now summarize:

\smallskip

\textbf{Application 1: Certainty Equivalents.} We evaluate two models on a set of binary lotteries from \citet{Bruhin}: a popular three-parameter specification of Cumulative Prospect Theory \citep{CPT}, henceforth CPT, and a two-parameter specification of Disappointment Aversion \citep{gul1991}, henceforth DA. We find that  CPT performs strikingly well on the \citet{Bruhin} data, achieving a completeness of 95\%, while DA's completeness is only 27\%.

One explanation for this finding is that CPT is a much better model of risk preferences than DA. Another possibility is that CPT is simply more flexible. We thus evaluate the restrictiveness of the two models, where our background constraints are that the synthetic average certainty equivalents must lie within the range of the lotteries' possible payoffs, and must respect first-order stochastic dominance (FOSD). We find that CPT is indeed substantially  less restrictive than DA: CPT performs better than DA not only on the real data set but also on the other eligible data sets. This tells us that FOSD constitutes a large part of the empirical content of CPT on the domain of binary lotteries, while DA imposes substantial additional restrictions.\footnote{DA's low completeness suggests that these restrictions are not supported  by the experimental 
data.}

Besides comparing distinct models such as CPT and DA, restrictiveness and completeness can be compared across nested models to reveal the role played by specific parameters. Adding a parameter always at least  weakly increases completeness and decreases restrictiveness, but some parameters achieve greater improvements in completeness for the same decrease in restrictiveness. We find that several parameters lead to large drops in restrictiveness in return for only marginal improvements in completeness, suggesting that these parameters may add flexibility in the wrong directions. The CPT parameter that governs the curvature of the probability weighting function, however, achieves a large improvement in completeness compared to the flexibility it adds, so this parameter seems to  capture an important part of risk preferences. Indeed, it is the curvature of the probability weighting function that has played a key role in many of the applications of CPT to financial data (e.g., \citet{barberis2008stocks} and \citet{green2012initial}).
\smallskip

\textbf{Application 2: Initial Play in Games.}
Next, we evaluate three models on a set of $3 \times 3$ matrix games from \cite{FudenbergLiang}: the Poisson Cognitive Hierarchy Model, or  \emph{PCHM} \citep{CamererHoChong04};  \emph{Logit PCHM} \citep{LeytonBrownWright}, which allows for logistic best replies in the PCHM; and \emph{Logit Level-1}, which models the distribution of play as a logistic best reply to the uniform distribution. We impose the background constraint that strictly dominant actions are played at least as often as if by chance (i.e. with probability at least $1/3$) and that strictly dominated actions are played with probability no more than $1/3.$ We find that all three models are highly restrictive relative to these constraints, which shows that the constraints on the frequency  of strictly dominated and strictly dominant strategies are a very small part of their empirical content. The restrictiveness of Logit PCHM and Logit Level-1 is nearly identical, although Logit PCHM has two parameters while Logit Level-1 has one.

\textbf{Application 3: Diffusion on a Social Network.} Finally, we consider the prediction of microfinance takeup rates in the set of Indian villages studied by \citet{banerjee2013diffusion, banerjee2019using}, and compare
the performance of OLS regression on various economically-motivated regressors with that of an economically-motivated partially linear model built upon ``network gossip centrality.'' Here we find that the partially linear model is  dominated by a simple OLS model based on the average eigenvector centrality of leaders: the latter has higher restrictiveness  and  higher completeness. 

\medskip

Besides these specific findings about each of these economic domains, our analyses make the high-level point that it is not sufficient to count parameters to understand a model's restrictiveness.  Even with just 3 parameters,  CPT is not very  restrictive  on the domain of binary lotteries, and models with different numbers of parameters (such as Logit PCHM and Logit Level-1) turn out to be similarly restrictive. These comparisons are not obvious from the functional forms, but are easy to discover with  our restrictiveness measure.

\section{Example}

Before formally defining our measure, we use a simple example to illustrate it. Suppose there is a binary covariate $x \in \{x_0,x_1\}$ and an outcome variable $y \in [0,1]$. A \emph{data set } is an observed outcome for each covariate value, i.e., a point in $\mathbb{R}^2$, and the \emph{eligible data} $\mathcal{F}$ is  collection of possible data sets, i.e. a subset of $\mathbb{R}^2$.  A \emph{model} is also a subset of $\mathbb{R}^2$. The model \emph{explains a data set exactly} if the data set is an element of the model.

\begin{figure}[H]
\begin{center}
\includegraphics[scale=0.5]{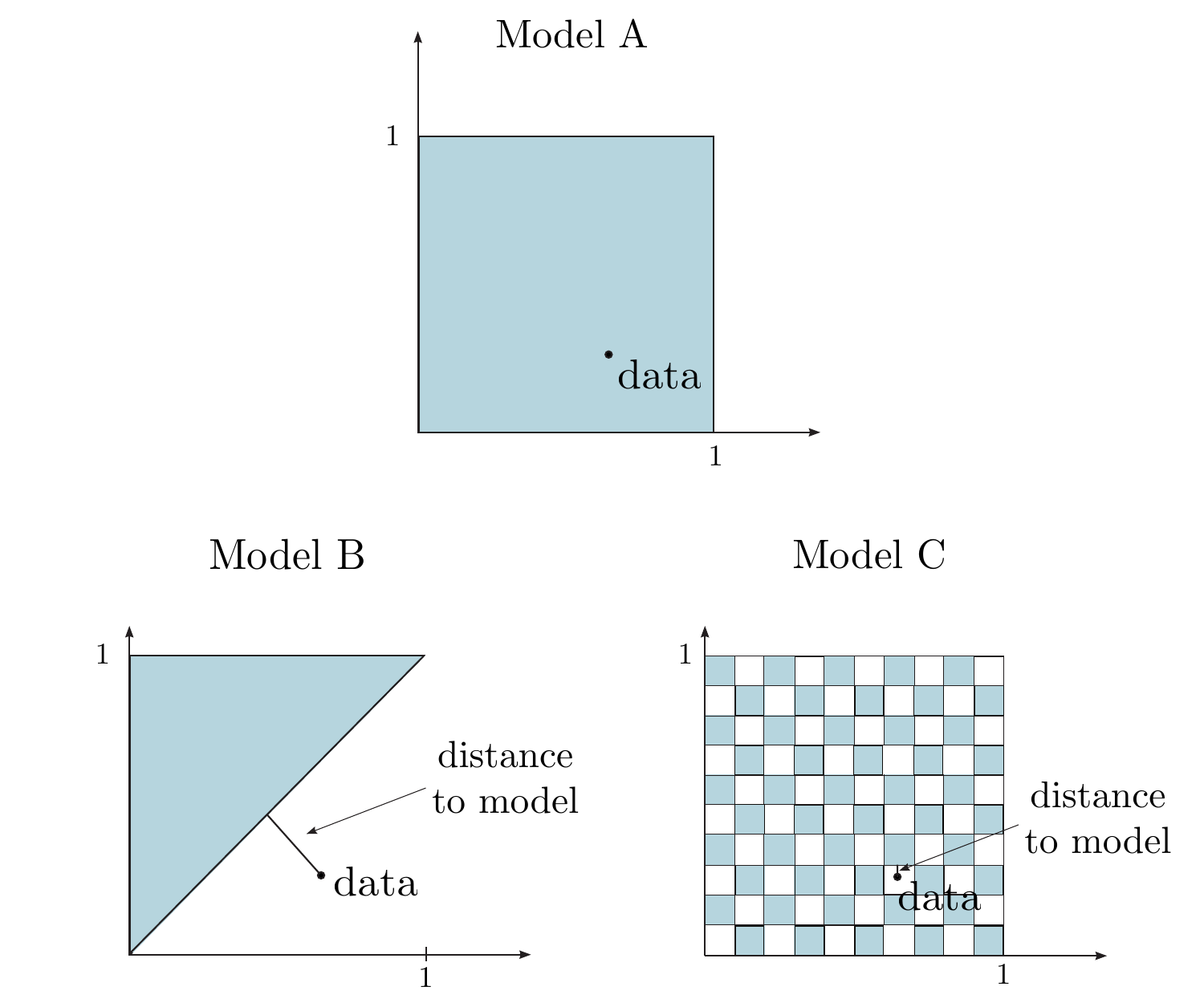}
\end{center}
 \caption{Three example models.}
\label{fig:ModelA_B}
\end{figure}

Figure \ref{fig:ModelA_B} considers eligible data $[0,1]^{2}$ and depicts three models. Model A includes all of $[0,1]^{2}$, and thus can exactly explain any (eligible) data set. Model B includes all data sets $(y_0,y_1)$ satisfying $y_1 > y_0$,  and so can only explain data sets where the outcome is higher at covariate $x_1$ than at $x_0$.  Model C discretizes the data into a grid and includes every other element of the grid.

One way of evaluating the restrictiveness of these models is the fraction of eligible data sets that they can fit exactly  \citep{Selten}. 
But evaluating restrictiveness in this way  obscures important differences between models such as B and C. Both exactly explain 50\% of the data yet Model B appears to impose a more substantive restriction. 

Our restrictiveness measure instead takes as given a measure of how well a model  approximates the  data. This  makes it computationally straightforward to estimate in  applications even when we have very little analytical guidance about the model's predictions (in contrast to the \citet{Selten} measure, which requires determining exact fit).   To estimate restrictiveness, we uniformly sample over all eligible data, evaluate the model's average approximation error to the realized datasets, and compare it to average approximation error of  a benchmark model.  For example, if  we use Euclidean distance as our measure of approximation error (as in Figure \ref{fig:ModelA_B}), and the constant model  $\{(1/2,1/2)\}$ as the  benchmark, the restrictiveness of Model B is numerically estimated to be about $0.30$, while the restrictiveness of Model C is approximately $0.02$. Thus  Model B is substantially more restrictive by our measure.

\section{Our Methodology}
Section  \ref{sec:Definition} formally defines our measure of restrictiveness. Section \ref{sec:completeness} reviews the measure of completeness from \citet{FKLM}. Section \ref{sec:Pareto} combines these concepts with the idea of a Pareto frontier of models that are undominated in completeness and restrictiveness. Section \ref{sec:Discussion} further discusses the interpretation of our restrictiveness measure. Section \ref{sec:related} describes the relationship to the literature. 

\subsection{Setup}
Our starting point is a data set of observations of $(X,Y)$, where $X$ is a \emph{covariate vector} and $Y \in \mathcal{Y}$ is an \emph{outcome}, with $\mathcal{Y}$ a compact subset of a finite-dimensional Euclidean space. We use $\mathcal{X}$ to denote the set of covariate vectors, and $P_X$ to denote the marginal distribution of $X$. We assume that $\mathcal{X}$ is finite, and $P_X$ is chosen by or known to the researcher.\footnote{In laboratory experiments the set of  features and their relative frequencies are chosen by the experimenter while in field experiments these are chosen by Nature, but in either case we treat them as known.} A \emph{prediction rule} is a function  $f: \mathcal{X} \rightarrow \mathcal{Y}$. We denote the set of all such functions  by $\ol{\mathcal{F}} \equiv \mathcal{Y}^{|\mathcal{X}|}$, and endow it with the usual topology. 

\begin{example}[Predicting an Average Outcome]
In our application to the prediction of  certainty equivalents (Section \ref{sec:App_CE}), the covariate vectors are 25 binary lotteries, each described by two prizes and their probabilities, and the outcome space is the observed average (over subjects) certainty  equivalent  for each lottery in this data set. A prediction rule is  any function from the 25 lotteries to average certainty equivalents.
\end{example}

\begin{example}[Predicting a Distribution]
In our application to initial play in 3x3 games (Section \ref{games}), the features are the 18 elements of the payoff matrix, and the outcomes are distributions over the row player's actions. A prediction rule is a map from payoff matrices to probability distributions over row player actions.
\end{example}

\subsection{Measures}
\subsubsection{ Restrictiveness} \label{sec:Definition}

We take as a primitive a  \emph{discrepancy} function $d: \ol{\mathcal{F}} \times \ol{\mathcal{F}} \rightarrow \mathbb{R}_+$ where $d(f,f')$ measures how different the two prediction rules $f$ and $f'$ are. For example, if $Y$ is a vector in $\mathbb{R}^{n}$,  a natural choice for $d$ is the expected mean-squared distance between the predictions (with respect to $P_X$), and if $Y$ is a distribution a natural choice for $d$ is the expected KL-divergence (again with respect to $P_X$).  We allow for functions $d$ that are not distances (such as KL-divergence), but  require that $d(f,f')=0$ if and only if $f=f'$. We also assume that $d$ is uniformly bounded, and that $d(\cd,f)$ and $d(f,\cd)$ are continuous almost everywhere for each $f\in\ol{\cF}$.\footnote{Given that $\cY$ is assumed to be bounded, the uniform boundedness of $d$ is a very weak requirement. The only reason that we allow for discontinuity in $d$ is to accommodate the case of $\ind\{f=f'\},$ the discrepancy function used in \cite{Selten}. We recommend in Appendix B that practitioners use a continuous discrepancy function $d$.}

We will evaluate the restrictiveness of a parametric model $\mathcal{F}_{\Theta}:=\{f_{\theta}\}_{\theta\in\Theta}\subseteq \overline{\mathcal{F}}$, where the prediction rules $f_{\theta}$ depend continuously on a parameter $\theta$ from a compact set $\Theta$.\footnote{Because $\mathcal{X}$ is assumed to be finite, $\Theta$ can viewed as a subset of a finite-dimensional Euclidean space  without loss of generality.} Restrictiveness is defined relative to a compact set of ``eligible'' rules $\mathcal{F} \subseteq \ol{\mathcal{F}}$ that reflect any constraints the model is known to have. For example, if a model is known to imply that choices respect first-order stochastic dominance, we can define $\mathcal{F}$ to be all rules with this property, and  measure the model's additional restrictiveness beyond this.  In general, the eligible set $\mathcal{F}$ consists of all prediction rules that satisfy user-specified background constraints, where the special case of  $\mathcal{F} = \ol{\mathcal{F}}$ corresponds to the question of whether $\mathcal{F}_{\Theta}$ imposes any restrictions at all.

We define the restrictiveness of a model to be   its expected discrepancy to a prediction rule $f$ drawn uniformly at random from the eligible set, normalized with respect to the expected discrepancy of a baseline prediction rule $f_{\text{base}}$. The baseline prediction rule is chosen to suit the setting, and we interpret its performance as a lower bound that any sensible model should outperform.\footnote{For example, in our application to predicting initial play in games, we define the baseline prediction rule to be a uniform distribution over actions. Note that while the choice of baseline affects the value of restrictiveness, it does not affect the comparative restrictiveness of two models on the same domain.}
\begin{definition} The restrictiveness of model $\mathcal{F}_{\Theta}$ with respect to eligible set $\mathcal{F}$ is \begin{equation} \label{eq:r}
 r(\mathcal{F}_{\Theta},\mathcal{F}) =   \frac{\mathbb{E}_{\lambda_{\mathcal{F}}}[ d(\mathcal{F}_{\Theta} ,f)] } {\mathbb{E}_{ \lambda_{\mathcal{F}}}[ d(f_{\text{base}},f) ]}
  \end{equation}
where $\lambda_{\mathcal{F}}$ denotes the uniform distribution  on $\mathcal{F}$,\footnote{Since $\mathcal{F}$ is a subset of bounded finite-dimensional Euclidean space, the uniform distribution on $\mathcal{F}$ is well-defined. Section \ref{sec:axiomatic} discusses a generalization to other distributions.} and $d(\mathcal{F}_\Theta,f) :=  \inf_{f_{\theta} \in \mathcal{F}_{\Theta}} d(f_{\theta},f)$.\footnote{When $\lambda_{\cal F}$ is interpreted as a Bayes prior, then restrictiveness can be interpreted as the ratio of Bayes risks defined with respect to the discrepancy function $d$. However, unlike in Bayesian statistics our goal is not to find a estimator whose ``Bayes risk'' is small. Indeed, a larger Bayes risk corresponds to higher restrictiveness, so all else equal we prefer models whose Bayes risk is higher.}
\end{definition}

 Normalizing with respect to a baseline has several advantages: First, it makes our  measure invariant to affine rescalings of the units of discrepancy. Second, whenever $f_{\text{base}}$ is chosen from $\mathcal{F}_{\Theta}$,  restrictiveness ranges from $0$ to 1.  A model with $r=0$ is completely unrestrictive, while a model with $r=1$ fits synthetic data no better than the baseline prediction rule does. If a model performs well on real data and is also highly restrictive, then its good performance occurs not simply  because the model can fit any data, but because it precisely identifies regularities in real behavior.

The ratio in (\ref{eq:r}) is well-defined as long as the denominator exceed zero, so we will impose this an assumption going forward:

\begin{assumption} \label{ass:denomNonZero}
$\E_{\lambda_{\mathcal{F}}}[d(f_{\text{base}},f)] > 0$.
\end{assumption}
 Section \ref{sec:axiomatic} provides axioms for the restrictiveness measure, which help to clarify the measure's theoretical properties.

\subsubsection{Completeness} \label{sec:completeness}

While restrictive models are desirable holding all else equal, a restrictive model is not  useful if it poorly fits real data. To evaluate model fit to  real data, we use the \emph{completeness} measure introduced in \citet{FKLM}. This takes as a primitive a \emph{loss function} $l: \mathcal{Y} \times \mathcal{Y} \rightarrow \mathbb{R}_+$, which is assumed to be  continuous. Let $P_{Y|X}$ denote the distribution of $Y$ given $X$, and $P:=(P_X, P_{Y|X})$ denote the joint distribution of $X$ and $Y$. The prediction rule that minimizes expected loss on the real data is given by
\[f^* \in \argmin_{f \in \overline{\mathcal{F}}} e_{P}(f)\] where
\[e_{P}(f):=\E_{P}\left[l(f(X),Y)\right] \quad \forall f\in \ol{\cal F}.\]
For example, if $\mathcal{X}$ is a set of lotteries,  $\mathcal{Y}$ are subjects' reported certainty equivalents for each lottery, and $l$ is squared error, then $f^*$ takes each lottery into its average certainty equivalent across subjects. If   $\mathcal{X}$ is a set of payoff matrices, $\mathcal{Y}$ is the set of distributions over actions,  and $l(Y,Y')$ is Kullback-Leibler divergence from $Y'$ to $Y$, then $f^*$ maps each game to the corresponding distribution over actions.

\begin{definition}[\citealp{FKLM}] \label{def:Completeness} The \emph{completeness} of model $\mathcal{F}_{\Theta}$ is defined by
\[
\kappa(\mathcal{F}_{\Theta}):=  \frac{e_{P}(f_{\text{base}})-\inf_{f\in\mathcal{F}_{\Theta}} e_{P}(f_{\theta})}{e_{P}(f_{\text{base}})-e_{P}(f^{*})}.
\]
\end{definition}

By construction,  $\kappa$ lies within the unit interval. A model with $\kappa=1$ matches the true $f^{*}$ exactly, while a model with $\kappa=0$ is no better at matching $f^{*}$ than the baseline prediction rule $f_{\text{base}}$. In the special case where discrepancy is the expected mean-squared distance $d(f,f') = \mathbb{E}_{P_X}[(f(X)-f'(X))^2]$ and the baseline prediction rule is constant at the expectation of $Y$,  $f_{\text{base}} = \E_{P}[Y]$, completeness specializes to the familiar (population) definition of $R^2$, but completeness is applicable more generally.

We report both restrictiveness $r$ and completeness $\kappa$ for each application that we consider. Completeness is defined using the loss function $l$, while restrictiveness is defined using the discrepancy function $d$.  When the discrepancy function $d$ and the loss function $l$ are ``paired'' in the sense of Online Appendix \ref{sec:Extend},\footnote{Loosely speaking, being paired means that  $d(f,f^*)$ is the difference between the error of $f$ and the error of the best mapping $f^*$.} 
then $\kappa(\mathcal{F}_{\Theta}) = 1 - r(\mathcal{F}_{\Theta},\overline{\mathcal{F}})$, so that completeness is the complement of the restrictiveness of model $\mathcal{F}_{\Theta}$ with respect to the (unconstrained) eligible set $\overline{\mathcal{F}}$.   Our first and third application use mean-squared error as the loss function and expected squared distance as the discrepancy function; our second application uses negative log-likelihood as the loss function and expected KL divergence as the discrepancy function. Both are examples of paired functions.

\subsubsection{A ``Pareto Frontier''} \label{sec:Pareto}
Our restrictiveness and completeness measures generate a  ``Pareto frontier'' consisting of models that are undominated in the sense that none of the  other models considered are simultaneously more restrictive and more complete. Although this is a very partial order, it has bite in our  Application \ref{sec:App_CE} (see  Figure \ref{fig:Risk}), as well as in the work of  \citet{EllisKarizOzbay}.

Unlike in typical economic problems, the Pareto frontier here  need not be concave, so the preferred model may not maximize a weighted sum of the two scores. For example, the frontier might consist of 3 points with scores (3/4,1/4), (1/3,1/3), and (1/4,3/4), and the analyst might prefer the model with scores $1/3$ each.  Of course, given the estimated parameter values of two models on the actual and hypothetical data sets, one could make predictions by taking pointwise combinations of the two model's predictions, which would mechanically lead to a weakly concave frontier of undominated models, but it seems hard to interpret this exercise.

While it is natural to prefer undominated models to dominated ones, it is less obvious how to aggregate the two measures to pick a preferred model, as the tradeoff between the measures is context-specific and also a matter of taste.  Nevertheless, when two models have completeness-restrictiveness values  that cannot be Pareto-ranked, one can consider the size of the improvement in completeness relative to the size of the reduction in restrictiveness. In Section \ref{subsec:Parameter} we show that adding an ``elevation''  parameter to a Cumulative Prospect Theory specification leads to  a large drop in restrictiveness in return for only
a small gain in completeness, while the 
parameter that governs the curvature of the probability weighting function leads to a sizeable improvement in completeness with only a small reduction in restrictiveness. We take this to mean that the curvature parameter  plays a more important role in capturing
risk preferences.\footnote{\citet*{ba2023over} conduct a similar exercise to compare two models which are not Pareto-ranked.}

\subsection{Discussion} \label{sec:Discussion}

\paragraph{Context dependence.} 
Restrictiveness is context-specific, in the sense that it depends on the set of feature vectors $\mathcal{X}$ and the outcome to be predicted. For example, we show that the restrictiveness of Cumulative Prospect Theory depends on the support size of the lotteries that are considered.  Evaluating the restrictiveness of a model across contexts can reveal that it   is very restrictive for one kind of prediction problem but unrestrictive for others.
An interesting direction for followup work would be to develop a measure of restrictiveness that takes into account how restrictive a model is across different contexts. For example, we might consider one model to be ``generally more restrictive'' than a second model if the distribution of restrictiveness values for the first model first-order stochastically dominates the distribution for the latter, as we find in Section \ref{sec:RobustMu}.

\paragraph{Choosing the eligible set.} The restrictiveness of a model is measured with respect to a specific eligible set $\mathcal{F} \subseteq \overline{\mathcal{F}}$, which is chosen based on what is known about the model. In Application 3, we investigate the restrictiveness of a structural model of network diffusion for predicting takeup of microfinance. Since there is relatively little known about the empirical content of this model, we define the eligible set to include all possible takeup rates, and study whether the model placed any restrictions at all.  In contrast,  the model of interest in Application 1,  Cumulative Prospect Theory, implies that any lottery that  first order stochastically dominates another must have a higher certainty equivalent. So we place this restriction on the eligible set, and see how much additional restrictiveness the model imposes. 

In general, there is not a single correct choice of eligible set. While we focus on  comparing the restrictiveness of models with respect to a given eligible set, an interesting complementary exercise is to fix a model and compare its restrictiveness relative to different eligible sets, as in Sections \ref{sec:RobustMu} and \ref{sec:RestrictivenessGames}.

\paragraph{Why the uniform distribution?} Section \ref{sec:axiomatic}, which  develops and axiomatizes a  broader class of restrictiveness measures, provides an axiom that pins down the uniform distribution.  Besides this axiom, there are many reasons to prefer the uniform distribution. First, once the eligible set is   specified, the uniform distribution on this set is pinned down (under our assumptions that $\mathcal{X}$ is finite and $\mathcal{Y}$ is a subset of finite-dimensional Euclidean space). This reduces the number of primitives to be chosen, and helps prevent cherry-picking with respect to the distribution on $\mathcal{F}$. Second, the uniform distribution is computationally easy to implement, even for eligible sets $\mathcal{F}$ with potentially complicated structures.\footnote{For example, in our  application to prediction of certainty equivalents, we build monotonicity with respect to FOSD into our definition of $\mathcal{F}$, and it is straightforward to sample uniformly from $\mathcal{F}$ by first sampling from a larger space without the monotonicity constraints, and then only keeping the draws that satisfy the monotonicity constraints. In contrast, non-uniform weightings over $\mathcal{F}$ require additional specification of how exactly $\mathcal{F}$ is parametrized, making the dependence of restrictiveness on $\mathcal{F}$ less transparent.}   Finally, our use of the uniform distribution follows up on \citet{Becker}'s proposal of the uniform distribution over budget-exhausting bundles as a model of irrational consumer behavior, and parallels \citet{Selten}'s use of area (see Section \ref{sec:related}).

\paragraph{Why are more restrictive models better?}

Our paper takes the perspective that restrictiveness is inherently desirable: if two models have the same level of predictive accuracy, we should prefer the one that  imposes more restrictions to the more flexible alternative. A potential reason for this preference is that models are  often meant to capture behavior in  related but not-identical domains. Given enough data, models that are very unrestrictive will fit any specific data set well, but may do so by learning idiosyncratic details of those datasets that do not in fact transfer across settings. In contrast, if a highly specific and structured model happens to fit a data set well, this  may generate more confidence that the model's structure extends to  other settings.\footnote{\citet{AFLW} compare the transfer performance of highly flexible black box models with less flexible economic models in a setting similar to our Application 1, and find that the black box models transfer more poorly.}

\subsection{Relationship to the Literature} \label{sec:related}

Our restrictiveness measure generalizes the notion of ``observational restrictiveness" introduced in \citet{Koopmans}, where a model is observationally restrictive if  the distributions permitted by the model are a proper subset of the distributions that would otherwise be possible.\footnote{As \citet{Koopmans} points out, a special case of an observationally restrictive specification is an overidentifying restriction. See e.g.  \cite{sargan1958estimation}, \cite{hausman1978specification}, \cite{hansen1982large}, and \cite{chen2018overidentification} for econometric tests of overidentification.} A model that is not observationally restrictive can perfectly match all data and so has $r=0$. Our restrictiveness measure allows us to quantify just how restrictive a model is.

Section 2 already discussed \citet{Selten}'s measure of flexibility,  and showed how its use of exact instead of approximate fit can lead to very different conclusions than ours.  The Selten measure has been  applied by \citet{BeattyCrawford}, \citet{Hey}, and \citet{HarlessCamerer}, and \citet{BlowBrowningCrawford} among others, to understand the restrictiveness of nonparametric economic models. It is typically difficult to determine  whether a parametric model can exactly fit a given data set without the guidance of prior analytical results, while  our measure is easy to compute in a variety of applications.\footnote{\citet{BeattyCrawford}  analytically derives the set of budget shares that are consistent with GARP, and \citet{HarlessCamerer} uses results about generalized expected utility theories to determine whether choices between specially chosen pairs of lotteries (for example, lotteries sharing a common ratio of outcome probabilities) are consistent with those theories. But we do not know how to analytically determine the predictions that are consistent with PCHM or the structural model of microfinance takeup in Application 3.}

In considering approximate rather than exact fit, our approach is related to papers that measure the distribution of the Afriat index \citep{Choietal,Polisson}.\footnote{\citet{Choietal} and \citet{Polisson} relax the implications of expected utility maximization using Afriat's ``efficiency index'' as an analog of our loss function. They compare the distribution of the efficiency indices of the actual subjects with its counterpart in randomly generated data.} These approaches are motivated by the testing of rationality of choices; our aim here is to show that similar techniques can be applied to a substantially broader class of models. \citet{BeattyCrawford} propose an alternative ``smoothed out" version of \citet{Selten}'s measure for the revealed preference setting that resembles restrictiveness, except that it does not allow for restrictions on the eligible data and  normalizes by reference to a worst case.\footnote{Another approach for model selection that does not require exact fit is \citet{ClippelRozen}'s suggestion to select  models by comparing the ratio of the likelihood of observing the real data under the specified model to the likelihood under a uniform distribution over all possible models.}

Our use of synthetic data to evaluate restrictiveness is similar to the use of simulated data to evaluate the power of a hypothesis test, as in \citet{Bronars} and \citet{andreoni2013power}. Their power measures are based on particular specifications of the alternative hypothesis, while we focus on an aggregate measure over a class of ``alternative hypotheses.'' Moreover, because  our objective is to measure the content of a model's restrictions and not hypothesis testing, we use approximate rather than exact fit.

Our measure is  related to various measures from computer science, statistics, and econometrics, but differs in a few key ways. First, compared to classic measures for the complexity of function classes, such as VC dimension, Rademacher complexity, and metric entropy,  our measure can be computed without analytical results about the empirical content of the estimated  model.

Second, compared to measures such as empirical Rademacher complexity, AIC, and BIC, which are often used for model selection, our restrictiveness measure does not depend on the observed data and is not indexed to sample size.\footnote{We could loosely interpret our restrictiveness measure as analogous to a limiting case of Rademacher complexity for large samples, where we use the  discrepancy function $d$, rather than correlation, to measure the model's ability to fit the synthetic data.} This reflects a difference in objectives: A primary goal of model selection is to avoid overfitting a complex model to a finite (and small) quantity of data, while our objective is to provide a measure of restrictiveness that does not depend on the quantity of data used to estimate it.\footnote{Specifically, our measure does not depend on the number of observations $(x,y)$ in the data or on the values of the $y$'s, though it does depend on the feature set $\mathcal{X}$.} Relatedly, while previous metrics aggregate a notion of completeness with some notion of restrictiveness,\footnote{For example, the AIC combines the log-likelihood, which is about fitness to real data (corresponding to ``completeness") and the number of parameters, which is about the flexibility of the model without reference to real data (corresponding to ``restrictiveness") in an additive way}  we trace the associated  Pareto frontier (see Section \ref{sec:Pareto}).

\section{Axiomatic Foundation for Restrictiveness} \label{sec:axiomatic}

This section provides an axiomatixation for the un-normalized version of the restrictiveness measure (i.e., the numerator of (\ref{eq:r})), which we call \emph{approximation error}.  Readers primarily interested in applications of the measure can skip ahead to the next section.

We endow the set $\overline{\mathcal{F}}$ with the Lebesgue $\sigma$-algebra and a $\sigma$-finite measure $\mu$, which can be interpreted as the analyst's prior. An approximation error $e$ takes as input the model $\mathcal{F}_{\Theta}\subseteq \overline{\mathcal{F}}$, a compact set of eligible prediction rules $\mathcal{F}\subseteq \overline{\mathcal{F}}$, and a  discrepancy function $d$. The quantity $e(\mathcal{F}_{\Theta},\mathcal{F},d)$ is interpreted as the approximation error of the model $\mathcal{F}_{\Theta}$ to the eligible set $\mathcal{F}$, where the quality of the approximation is measured using $d$. 
We would like for this approximation error function to satisfy the following axioms. First, approximation error should always be nonnegative.

\begin{axiom}[Nonnegativity] \label{axiom:Nonnegative} For every model $\mathcal{F}_\Theta$, eligible set $\mathcal{F}$, and discrepancy $d$, $e(\mathcal{F}_\Theta,\mathcal{F},d)\geq 0$.\end{axiom}

Second, if one model is better able to approximate every eligible prediction rule than another, the first model has lower approximation error. 

\begin{axiom}[Monotonicity] \label{axiom:Monotone} Fix any set of eligible mappings $\mathcal{F}$. If the sets $\mathcal{F}_{\Theta_1}$ and $\mathcal{F}_{\Theta_2}$ satisfy
 $d(\mathcal{F}_{\Theta_1},f) \geq d(\mathcal{F}_{\Theta_2},f)$ for all $f \in \mathcal{F}$, 
  then $e(\mathcal{F}_{\Theta_1},\mathcal{F},d) \geq e(\mathcal{F}_{\Theta_2},\mathcal{F},d)$.
 \end{axiom}

 Third, any linear rescaling of the units of $d$ is inherited  by the approximation error, and a linear rescaling of the discrepancy between a model $\mathcal{F}_{\Theta}$ to each prediction rule $f$ leads to the same value of approximation error as rescaling the units of the discrepancy $d$.

 \begin{axiom}[Homogeneity] \label{axiom:Rescale} (a)  Fix any model $\mathcal{F}_\Theta$, set of eligible prediction rules $\mathcal{F}$, and discrepancy $d$. Then
 $e(\mathcal{F}_\Theta,\mathcal{F},\alpha \cdot d) = \alpha \cdot e(\mathcal{F}_\Theta,\mathcal{F},d)$ for every $\alpha \in \mathbb{R}_+$\\
 (b) Fix any set of eligible prediction rules $\mathcal{F}$ and discrepancy $d$. If $\mathcal{F}_{\Theta_1}$ and $\mathcal{F}_{\Theta_2}$ satisfy
 $d(\mathcal{F}_1,f) = \alpha \cdot d(\mathcal{F}_2,f)$ for all $f\in \mathcal{F},$
 then $e(\mathcal{F}_{\Theta_1},\mathcal{F},d) = e(\mathcal{F}_{\Theta_2},\mathcal{F}, \alpha \cdot d).$
 \end{axiom}

Fourth, consider constraining the set of eligible prediction rules $\mathcal{F}$ to a subset $\mathcal{F}_1$ or its complement $\mathcal{F}_2$. The \emph{ex post} approximation errors of a model $\mathcal{F}_{\Theta}$ with respect to either of these new eligible sets is, respectively, $e(\mathcal{F}_{\Theta},\mathcal{F}_1, d)$ or $e(\mathcal{F}_{\Theta},\mathcal{F}_2, d)$. The subsequent axiom says that the ex ante approximation error $e(\mathcal{F}_\Theta,\mathcal{F},d)$ is a convex combination of the ex post approximation errors, where each ex post subset contributes to the ex ante approximation error in proportion to its measure.

 \begin{axiom}[Linearity] \label{axiom:Convexity}
  For any sequence of disjoint measurable sets $\mathcal{F}_{\Theta_1},\mathcal{F}_{\Theta_2}, \dots $ whose union $\mathcal{F}_\Theta \equiv \cup_{i=1}^\infty \mathcal{F}_{\Theta_i}$ has strictly positive measure,
 \[e(\mathcal{F}_\Theta,\mathcal{F},d) = \sum_{i=1}^\infty \frac{\mu( \mathcal{F}_{\Theta_i})}{\mu(\mathcal{F})} \cdot e(\mathcal{F}_{\Theta_i},\mathcal{F}, d) \quad \forall \mathcal{F},d.\]
 \end{axiom}

Finally, permuting the various discrepancies between the model and the eligible prediction rules $f$ does not affect the overall approximation error. This reflects a ``principle of indifference" over the eligible prediction rules.

 \begin{axiom}[Symmetry] \label{axiom:Symmetry} Fix any eligible set $\mathcal{F}$ and any bijection $\tau$ from $\mathcal{F}$ to itself.  Consider two sets $\mathcal{F}_{\Theta_1}$ and $\mathcal{F}_{\Theta_2}$ where
 $d(\mathcal{F}_{\Theta_1},f) = d(\mathcal{F}_{\Theta_2},\tau(f))$ for all $ f \in \mathcal{F}.$ Then $e(\mathcal{F}_{\Theta_1},\mathcal{F},d) =e(\mathcal{F}_{\Theta_2},\mathcal{F},d).$
 \end{axiom}

\begin{proposition} \label{prop:axiom} An approximation error $e$ satisfies Axioms \ref{axiom:Nonnegative}-\ref{axiom:Convexity} if and only if there is a function $c: \overline{\mathcal{F}} \rightarrow \mathbb{R}$ such that 
\begin{equation} \label{eq:e}
e(\mathcal{F}_\Theta,\mathcal{F},d) =   \mathbb{E}_{f \sim \mu_{\mathcal{F}}} \left[c(f) \cdot \inf_{g \in \mathcal{F} } d(g,f) \right] \quad \forall \mathcal{F}_\Theta,\mathcal{F},d
\end{equation}
where $\mu_{\mathcal{F}}$ denotes the measure $\mu$ conditional on the event $\mathcal{F}$. If additionally $e$ satisfies Axiom \ref{axiom:Symmetry}, then 
\begin{equation} \label{e:Uniform}
e(\mathcal{F}_\Theta,\mathcal{F},d) =   \mathbb{E}_{f \sim \lambda_{\mathcal{F}}} \left[\inf_{g \in \mathcal{F} } c \cdot d(g,f) \right] \quad \forall \mathcal{F}_\Theta,\mathcal{F},d
\end{equation}
for a positive constant $c$, where $\lambda$ denotes the Lebesgue measure on $\overline{\mathcal{F}}$. 
\end{proposition}

Our restrictiveness measure assumes (\ref{e:Uniform}), and normalizes the approximation error of  model $\mathcal{F}$ relative to the approximation error of the baseline  $f_{\text{base}}$.

\section{Computation and Estimation}\label{Est}

We now discuss how to implement our approach in practice. Recall that we restrict $\mathcal{X}$ to be finite, so $\overline{\mathcal{F}}$ is finite-dimensional.

\paragraph{Computing Restrictiveness}

The following is an algorithm for computing $r$: Sample $M$ times independently from a uniform distribution on the eligible set $\mathcal{F}$. For each sampled $f_m\in\cF$, compute $d(\mathcal{F}_{\Theta},f_m)$ and $d(f_{\text{base}},f_m)$.
Then
$$\hat{r}_M := \frac{\frac{1}{M}\sum_{m=1}^{M}d(\mathcal{F}_{\Theta},f_m)}{\frac{1}{M}\sum_{m=1}^{M}d(f_{\text{base}},f_m)}$$
is an estimator for restrictiveness $r = r(\mathcal{F}_{\Theta},\mathcal{F})$. In principle, the number of simulations we run, $M$, can be arbitrarily large, so $\hat{r}_{M}$ can be made arbitrarily close to $r$. Moreover, it is straightforward to obtain the formula for the asymptotic standard error of the simulated $r$, based on which confidence intervals can be constructed.\footnote{Under Assumption \ref{ass:denomNonZero},
   $\sqrt{M}\left(\hat{r}_{M}-r\right)/\hat{\s}_{\hat{r}}\dto\cN\left(0,1\right)$,
   where the asymptotic variance estimator $\hat{\s}_{\hat{r}}^2$ is defined by
$\hat{\s}_{\hat{r}}^{2}:=\left[\hat{\s}_{\cG}^{2}-2\hat{r}\hat{\s}_{\cG,f_{\text{base}}}+\hat{r}^{2}\hat{\s}_{f_{\text{base}}}^{2}\right]/\left[\left(\frac{1}{M}\sum_{m=1}^{M}d(f_{\text{base}},f_m)\right)^{2}\right]$, with $\hat{\s}_{\cG}^{2}$ being the sample variance of $d(\cG,f_m)$, $\hat{\s}_{f_{\text{base}}}^{2}$ the sample variance of $d(f_{\text{base}},f_m)$, and $\hat{\s}_{\cG,f_{\text{base}}}^{2}$ the sample covariance of $d(\cG,f_m)$ and $d(f_{\text{base}},f_m)$, across $m=1,...,M$. We note that the standard error here simply measures the approximation error of $r$ based on a finite number of simulations and do not reflect randomness in experimental data.}

\paragraph{Estimating Completeness} \label{sec:estimateComplete}

Suppose that the analyst has access to a finite sample of data $\left\{ Z_{i}:=\left(X_{i},Y_{i}\right)\right\} _{i=1}^{N}$ drawn from the unknown true distribution $P^{*}$. To estimate completeness, which is defined based on the loss function $l$ introduced in Section \ref{sec:completeness}, we use $K$-fold cross-validation to estimate the out-of-sample
prediction error of the model.

(Our applications make the standard choice of $K=10$.)
Specifically, we randomly divide ${\bf Z}_{N}=(Z_1, \dots, Z_N)$ into $K$ (approximately)
equal-sized groups. To simplify notation, assume that $J_{N}=\frac{N}{K}$
is an integer. Let $k\left(i\right)$ denote the group number of observation
$Z_{i}$, and fix an arbitrary set of maps  $\widetilde{\cal F}$. In the $k$-th fold of cross-validation, we will use the observations in group $k$ for testing and the remaining observations for training. 

For each group $k=1,...,K$, define $\hat{f}^{-k} :=\arg\min_{f\in \widetilde{{\cal F}}}\frac{1}{N-J_{N}}\sum_{k\left(i\right)\neq k}l(f,Z_i)$ to be the minimizer in $\widetilde{\cal F}$ on the $k$-th training set (i.e., all observations  outside of group $k$), and $\hat{e}_{k}  :=\frac{1}{J_{N}}\sum_{k\left(i\right)=k}l\left(\hat{f}^{-k},Z_i\right)$ to be the out-of-sample error on the $k$-th test set. Then the average test error across the $K$ folds, 
$\hat{e}_{CV}\left(\widetilde{{\cal F}}\right)  :=\frac{1}{K}\sum_{k=1}^{K}\hat{e}_{k}$, is an estimator for the unobservable expected error of the best prediction rule from class $\widetilde{\mathcal{F}}$. Setting $\widetilde{\cal F}$ to be $\overline{\mathcal{F}}$, $\cG$, or $\{f_{\text{base}}\}$, we can compute  $\hat{e}_{CV}\left(\overline{\mathcal{F}} \right)$, $\hat{e}_{CV}\left(\cG\right)$ and $\hat{e}_{CV}\left(f_{\text{base}}\right)$ from the data, leading to the following estimator for $\kappa$:
\[
\hat{\kappa}= 1 - \frac{\hat{e}_{CV}\left(\cG\right)-\hat{e}_{CV}\left({\cF^*}\right)}{\hat{e}_{CV}\left(f_{\text{base}}\right)-\hat{e}_{CV}\left({\cF^*}\right)}.
\]

It is crucial that the denominator in $\hat{\kappa}$ does not vanish
asymptotically, so we impose the following assumption:
\begin{assumption}[Baseline is Imperfect] \label{assu:naive}
  $e_{P}(f_{\text{base}})-e_{P}(f^*) > 0$.
\end{assumption}
This assumption says that the baseline  prediction rule performs strictly worse in expectation than the best prediction rule so there is some room for a model to do better. We show that $\hat{\kappa}$ is asymptotically
normal by adapting Proposition 5 in \cite{austern2020asymptotics}.

\begin{proposition}
\label{ANorm_kstar} Under Assumption \ref{assu:naive} and some regularity conditions,\footnote{See Appendix \ref{sec:asym_CV} for details of these assumptions.}
$\sqrt{N}\left(\hat{\kappa}-\kappa\right)/{\hat{\s}_{\hat{\kappa}}}\dto\cN\left(0,1\right)$,
where the variance estimator $\hat{\s}^2_{\hat{\kappa}}$ is as  defined in Appendix \ref{app:VarEst}.
\end{proposition}

\section{Application 1: Certainty Equivalents \label{sec:App_CE}}

\subsection{Setting}

Our first application is to the prediction of  certainty equivalents for a set of 25 binary lotteries from \citet{Bruhin}. Each lottery is described as a tuple $x=(\overline{z},\underline{z},p)$, where $\overline{z}>\underline{z} \geq 0$  are the possible prizes,  and $p$ is the probability of the larger prize. Each observation consists of a lottery and a reported certainty equivalent by a given subject, so we can describe the feature space $\mathcal{X}$ by the 25 lottery tuples $(\overline{z},\underline{z},p)$ in the \citet{Bruhin} data, and the outcome space by $\mathcal{Y}=\mathbb{R}$. Note that the residual uncertainty in $Y$ conditional on $X$ reflects heterogeneity in certainty equivalents reported across subjects for the same lottery.

We predict the average  certainty  equivalent  (over subjects) for each lottery in this data set.  A prediction rule for this problem is any function $f: \mathcal{X} \rightarrow \mathbb{R}$ from the 25 lotteries to their average certainty equivalents, and the discrepancy between two mappings is defined to be their average mean-squared distance 
$d(f,f') = \frac{1}{\vert \mathcal{X} \vert} \sum_{x \in \mathcal{X}} (f(x)-f'(x))^2.$
 
We evaluate the restrictiveness and completeness of two economic models. First we consider a three-parameter version of  \emph{Cumulative Prospect Theory} indexed by $\theta=(\alpha,\gamma,\delta)$, which specifies a utility  
$ w(p)v(\overline{z}) + (1-w(p)) v(\underline{z})$
for each lottery $(\overline{z},\underline{z},p)$, where
\begin{equation} \label{eq:utility}
v(z)= z^{\alpha},\quad w(p)= \frac{\delta p^\gamma}{\delta p^\gamma + (1-p)^\gamma}.
\footnote{This parametric form for $w(p)$ was used by \citet{goldstein1987expression} and
\citet{lattimore1992influence}. } 
\end{equation}

The predicted certainty equivalent of a binary lottery is then given by 
$f_\theta(\overline{z}, \underline{z},p) =v^{-1}\left( w(p)v(\overline{z}) + (1-w(p)) v(\underline{z})\right).$ 
Following the literature, we restrict $\alpha,\gamma \in [0,1]$, and $\delta\geq 0$. We specify  $\mathcal{F}$ as the set of all such functions $f_{\t}$ with parameters $\t$ in this range, and refer to this model simply as CPT. As a baseline, we consider the function $f_{\text{base}}$ that maps each lottery into its expected value, corresponding to $\alpha=\gamma=\delta=1$.

Second, we consider the \emph{Disappointment Aversion} model of \citet{gul1991}, using a parametrization proposed in \citet{routledge2010generalized}  with the parameters $\lambda=(\alpha,\eta)$,  where $\alpha \in [0,1]$ and $\eta >-1$.\footnote{To facilitate comparison with CPT, we depart slightly from \citet{routledge2010generalized} by imposing the functional form $v(z)=z^\alpha$ instead of $v(z) = z^\alpha/\alpha$.} The value function for money is the same as in (\ref{eq:utility}), but the probability weighting function is given instead by
$\widetilde{w}(p) = \frac{p}{1+(1-p)\eta}.$ There are two parameters: $\alpha$ again reflects the curvature of the utility function, while  $\eta>0$ corresponds to  ``disappointment aversion,'' i.e. aversion to  realizations of the lottery that are worse than its certainty equivalent.  Here the predicted certainty equivalent is
$f_\lambda(\overline{z},\underline{z},p)=v^{-1}(\widetilde{w}(p)v(\overline{z}) + (1-\widetilde{w}(p))v(\underline{z})).$  
 \noindent We specify  $\mathcal{F}_{\Lambda}$ as the set of all such functions and refer to this model as DA. Again, we use expected value as the  baseline prediction, which corresponds to $\alpha=1$ and $\eta=0$ in DA.

\subsection{Completeness}

We evaluate completeness using mean-squared error as the loss function, i.e., if the reported certainty equivalent is $y$ when the model predicts $\hat{y}$, the loss in that observation is $(\hat{y}-y)^2$.\footnote{This loss function is paired to the average mean-squared discrepancy function we used for measuring restrictiveness, see Appendix \ref{sec:Extend} for details.} CPT achieves a striking out-of-sample performance for 
predicting certainty equivalents in  the \citet{Bruhin} data: it is 95\% complete.\footnote{\citet{FKLM} reports a similar finding for a sample of gain-domain and loss-domain lotteries.} Thus, the model achieves almost all of the possible improvement in prediction accuracy over the baseline.\footnote{This finding is consistent with  \citet{PeysakhovichNaecker}'s result that CPT approximates the predictive performance of lasso regression trained on a high-dimensional set of features.} In contrast, DA is only 27\% complete on the same data. One explanation is  that CPT more precisely captures the observed risk preferences in the data than DA, but another possibility is that CPT is flexible enough to mimic most functions from binary lotteries to certainty equivalents, while DA imposes more substantial restrictions. These explanations have very different implications for how to interpret CPT's empirical success compared to DA's.

\subsection{Restrictiveness} \label{sec:CPTr}

To distinguish between these explanations, we now compute the restrictiveness of the two models. We define the eligible set to be all prediction rules satisfying the following criteria:
\begin{enumerate}
    \item[(i)] $\underline{z} \leq f(\overline{z},\underline{z},p) \leq \overline{z}$;
    \item[(ii)] If $\overline{z} \geq \overline{z}'$, $\underline{z}\geq \underline{z}'$, and $p\geq p'$ with at least one ``$\geq$" strict, then $f(\overline{z},\underline{z},p)> f(\overline{z}',\underline{z}',p').$
\end{enumerate}
Constraint (i) requires that the certainty equivalent is within the range of the possible payoffs, while (ii) is equivalent to monotonicity with respect to first-order stochastic dominance.\footnote{The CDF of a binary lottery  with $\ol{z}>\ul{z}$ and $0<p<1$ is $F(z) = (1-p)\ind\{ \ul{z} \leq z < \ol{z}\} + \ind\{z\geq \ol{z}\}$, which is weakly decreasing in $(\ol{z},\ol{z},p)$ for all $z$, so $(\ol{z},\ol{z},p)$ FOSD $(\ol{z}',\ol{z}',p')$ if and only if $(\ol{z},\ol{z},p) \gneqq (\ol{z}',\ol{z}',p').$ There are many pairs of lotteries in the \citet{Bruhin} lottery data that can be compared via (ii), so these conditions are not vacuous.}

Table \ref{tab:CPTDA} reports the completeness and restrictiveness of both models.

\begin{table}[H]

\centering{}%
\begin{tabular}{@{}l@{}lcc@{}}
\hline 
 & \# Param & \multicolumn{1}{c}{Restrictiveness} & \multicolumn{1}{c}{Completeness}\tabularnewline
\hline 
\hline 
 CPT & $\quad$$\quad$ 3 &  0.28 & 0.95 
 \tabularnewline
 & &  (0.003) & (0.02)     \tabularnewline
DA & $\quad$$\quad$ 2 & 0.47 &   0.27 \tabularnewline
&  & (0.006) & (0.06)  \tabularnewline
\hline 
\end{tabular}
  \caption{\footnotesize{Completeness for both models is estimated on the real data, which includes reported certainty equivalents by each of 179 subjects. Standard errors for the completeness estimates are computed using a block bootstrapping procedure that clusters together all observations from the same subjects, see Appendix \ref{app:table}. Restrictiveness is estimated from 1000 simulations.}} \label{tab:CPTDA}
\end{table}

The restrictiveness of CPT is $0.28$, so on average CPT's approximation error is about one fourth of the error of the expected value.  DA is more restrictive, with an average approximation error almost one half of the error of the baseline.  Thus the two models are not directly comparable: CPT performs substantially better for predicting the real data, but would have performed well out-of-sample given sufficient data from almost any underlying data-generating process that respects first-order stochastic dominance. DA rules out more behaviors that satisfy first-order stochastic dominance, but in doing so is unable to well approximate the actual \citet{Bruhin} data.

\subsection{The Role of a Parameter} \label{subsec:Parameter}
    
In addition to comparing  models such as CPT and DA, our approach can be used to learn more about the role played by specific parameters.
Adding a parameter must at least weakly decrease restrictiveness and increase completeness, but we find that parameters can differ substantially in  their effectiveness in trading off between these two goals. We also show that models with the same number of parameters can have very different levels of restrictiveness, and thus a simple parameter count is substantively less informative than our measure.

Specifically, we consider alternative specifications of CPT and DA with fewer  parameters. Some of these specifications have been studied in the literature: CPT($\alpha,\gamma$), with $\delta=1$, is used in \citet{karmarkar1979}\footnote{This specification with weighting function $w(p)=\frac{p^\gamma}{p^\gamma+(1-p)^\gamma}$ is very similar to one used in \citet{CPT}, where the weighting function was $w(p) = \frac{p^\gamma}{p^\gamma+(1-p)^\gamma)^{1/\gamma}}$.}; CPT($\gamma,\delta$), with $\alpha=1$, corresponds to a risk-neutral CPT agent whose utility over money is $u(z)=z$ but exhibits nonlinear probability weighting; CPT($\alpha$), with   $\delta=\gamma=1$, corresponds to an Expected Utility decision-maker whose utility function is as given in (\ref{eq:utility}), and is also equivalent to DA($\alpha$).\footnote{See the survey \citet{fehrdudaepper} for further discussion of these  different parametric forms, and others which have been used in the literature.} The model CPT($\gamma$), with $\alpha=\delta=1$, and CPT($\delta$), with $\alpha=\gamma=1$ have not been studied in the literature, but we report them for comparison. We also consider DA($\eta$) as in \citet{gul1991}, with $\alpha=1$, which corresponds to a disappointment-averse  decision maker whose utility is linear in  money.

        \begin{figure}[h] 
         \centering
         \includegraphics[scale=0.45]{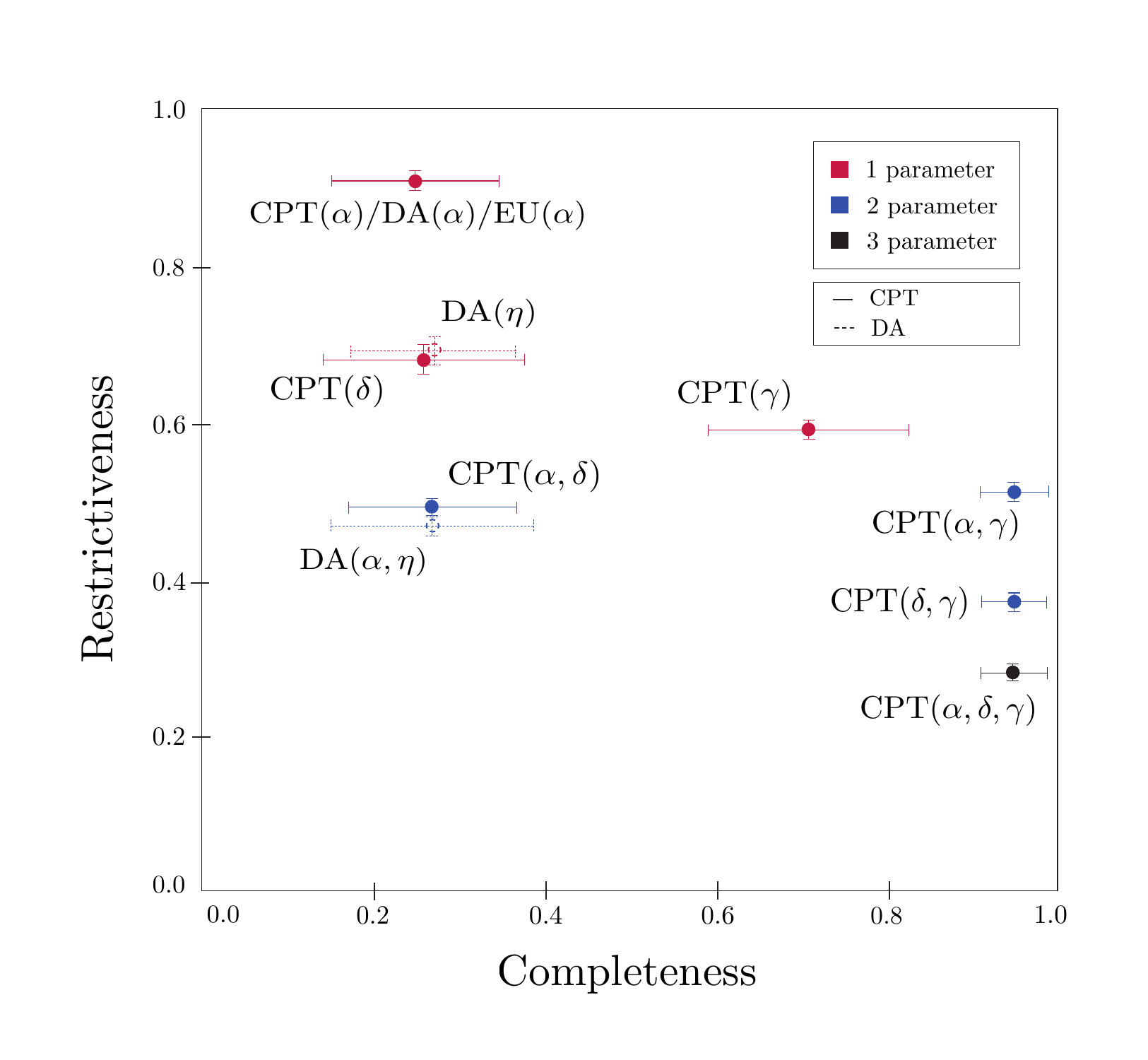}
           \caption{\footnotesize{Comparison of models by their completeness and restrictiveness.}} \label{fig:Risk}
         \end{figure}

     Figure \ref{fig:Risk} plots restrictiveness and completeness for these alternative specifications, which reveals that some  specifications fall in the interior of the restrictiveness-completeness Pareto frontier introduced in Section \ref{sec:Pareto}: Each of CPT($\alpha,\delta$) and DA($\alpha,\eta$) are dominated, in the sense that another model is simultaneously more complete and also more restrictive.\footnote{Each of  CPT($\alpha,\delta$) and DA($\alpha,\eta$) is less complete and less restrictive than the single parameter model CPT($\gamma$), and these differences are statistically significant. (See also Table \ref{tab:gains} in Online Appendix \ref{app:table}.)} The figure also reveals substantial dispersion in the restrictiveness of these specifications (ranging from $r=0.28$ to $r=0.92$), even though  all of the specifications use only a small number of parameters. This observation  emphasizes the distinction between our method and a simple  parameter count. 
     
 By looking more specifically at how restrictiveness and completeness vary across two nested specifications, we can better understand the role that any specific parameter plays. Figure \ref{fig:role-parameters} shows that the different parameters for probability weighting are not equally effective. Adding the  parameter $\delta$, which governs the elevation of the probability weighting curve, to any specification of CPT leads to a large drop in restrictiveness in return for only a small gain in completeness. We find a similar result for the ``disappointment aversion" parameter $\eta$ in DA, which barely improves upon the completeness of DA$(\alpha)$, but leads to a substantial drop in restrictiveness.     
In contrast, the parameter $\gamma$, which governs the curvature of the probability weighting function, appears to play an important role in capturing risk preferences: Adding $\gamma$ to any CPT specification leads to a sizeable improvement in completeness at the cost of a modest reduction in restrictiveness. This supports previous  findings that  probability distortions play an important role in fitting experimental and  field data  \citep{SnowbergWolfers,fehrdudaepper,Donoghue}.

        \begin{figure}[h] 
         \centering
         \includegraphics[scale=0.25]{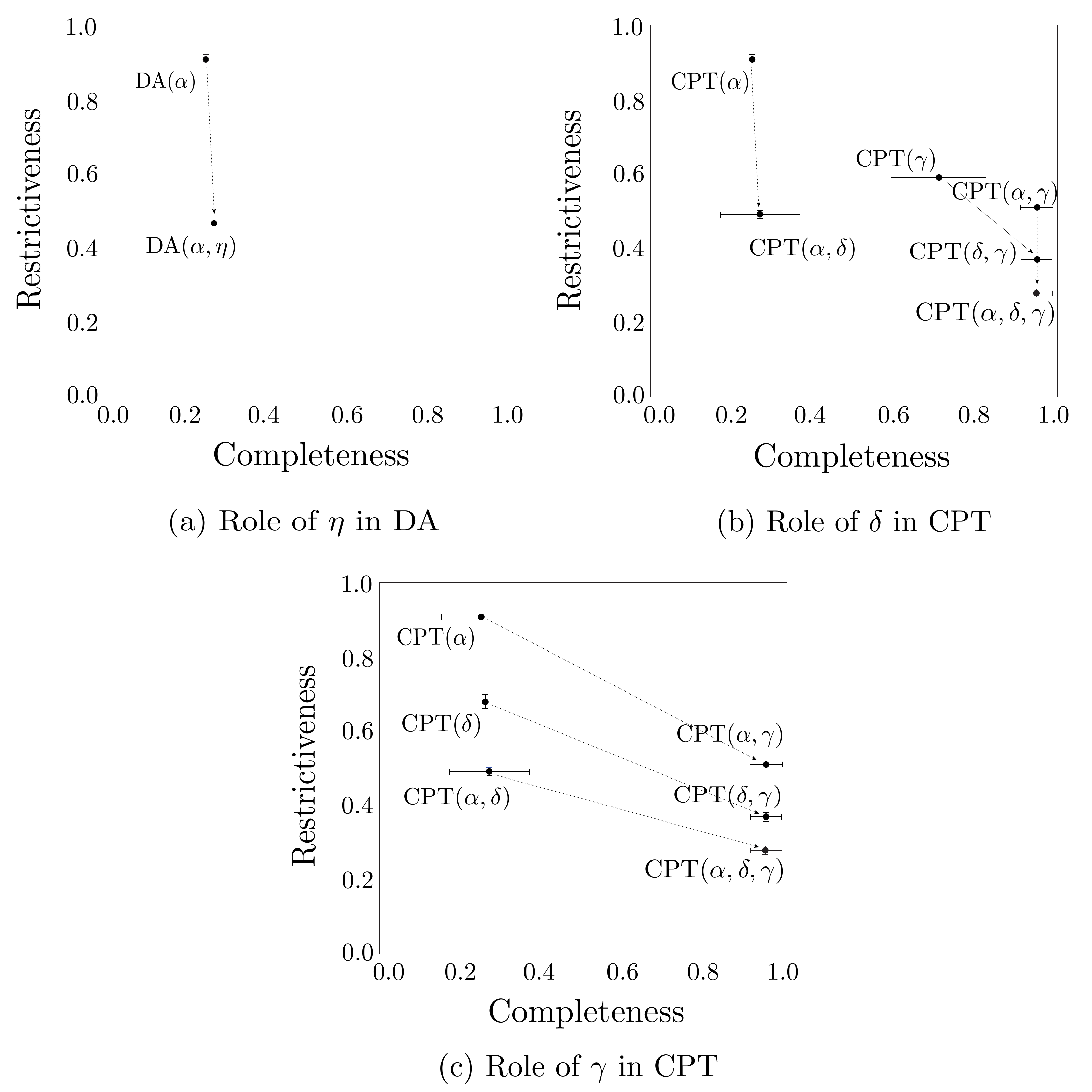}
           \caption{\footnotesize{Impact of the probability weighting parameters on completeness and  restrictiveness. }} \label{fig:role-parameters}
         \end{figure}

\subsection{Robustness Checks} \label{sec:RobustMu}

We show that the qualitative findings in this section are robust to certain natural changes in the eligible set and the feature set. Together with the robustness check in Section \ref{sec:RestrictivenessGames}, these results  also speak to the sensitivity of the restrictiveness measure in general: although the measure will typically vary with these specifications, it may not be very sensitive  in practice for many economic models of interest.

\paragraph{Different distribution over the eligible set.} The uniform distribution is the same as  $\mbox{beta}(1,1)$, so to test the sensitivity of the restrictiveness measure we consider nearby $\mbox{beta}(a,b)$ distributions with parameters $(a,b)$ sampled from a uniform distribution over $[0.9,1.1]\times [0.9,1.1]$. For each $(a,b)$ pair, we generate certainty equivalents from a  $\mbox{beta}(a,b)$ distribution over the prize range, again keeping only those functions $f$ that satisfy FOSD. Over 100 such distributions $\mbox{beta}(a,b)$, the average restrictiveness is 0.29, with a minimum value of 0.27 and a maximum value of 0.32.

\paragraph{Different eligible set.} Next, we compute the restrictiveness of  CPT$(\alpha,\delta,\gamma$) with respect to an eligible set that imposes the range restriction in (i) but drops the FOSD restrictions in (ii).  The model's errors are substantially higher when we drop FOSD (increasing from 63.75 to 102.41), but so are the errors of the  Expected Value benchmark. The relative performance  of CPT$(\alpha,\delta,\gamma)$ compared to the expected-value baseline is nearly identical regardless of whether or not we impose FOSD:  the model's restrictiveness relative to this larger eligible set is 0.29 (compared to  0.28 relative to the original eligible set). 

\paragraph{Other sets of binary lotteries.} In our main analysis, the feature space $\mathcal{X}$  consisted of 25 binary lotteries from \citet{Bruhin} data. Below we report  the restrictiveness of CPT$(\alpha,\gamma,\delta)$ and DA$(\alpha,\eta)$ with respect to alternative sets of binary lotteries, drawn from five additional papers (see Appendix \ref{sec:TableROther} for details). Figure \ref{fig:CDF} shows the CDF of restrictiveness values across these lotteries (including  the  \citet{Bruhin} lotteries) for both models. We find that CPT is not very restrictive on any of these sets of lotteries, and that the distribution of restrictiveness for DA first-order stochastically dominates that of CPT.

\begin{figure}[H]
    \centering
    \includegraphics[scale=0.15]{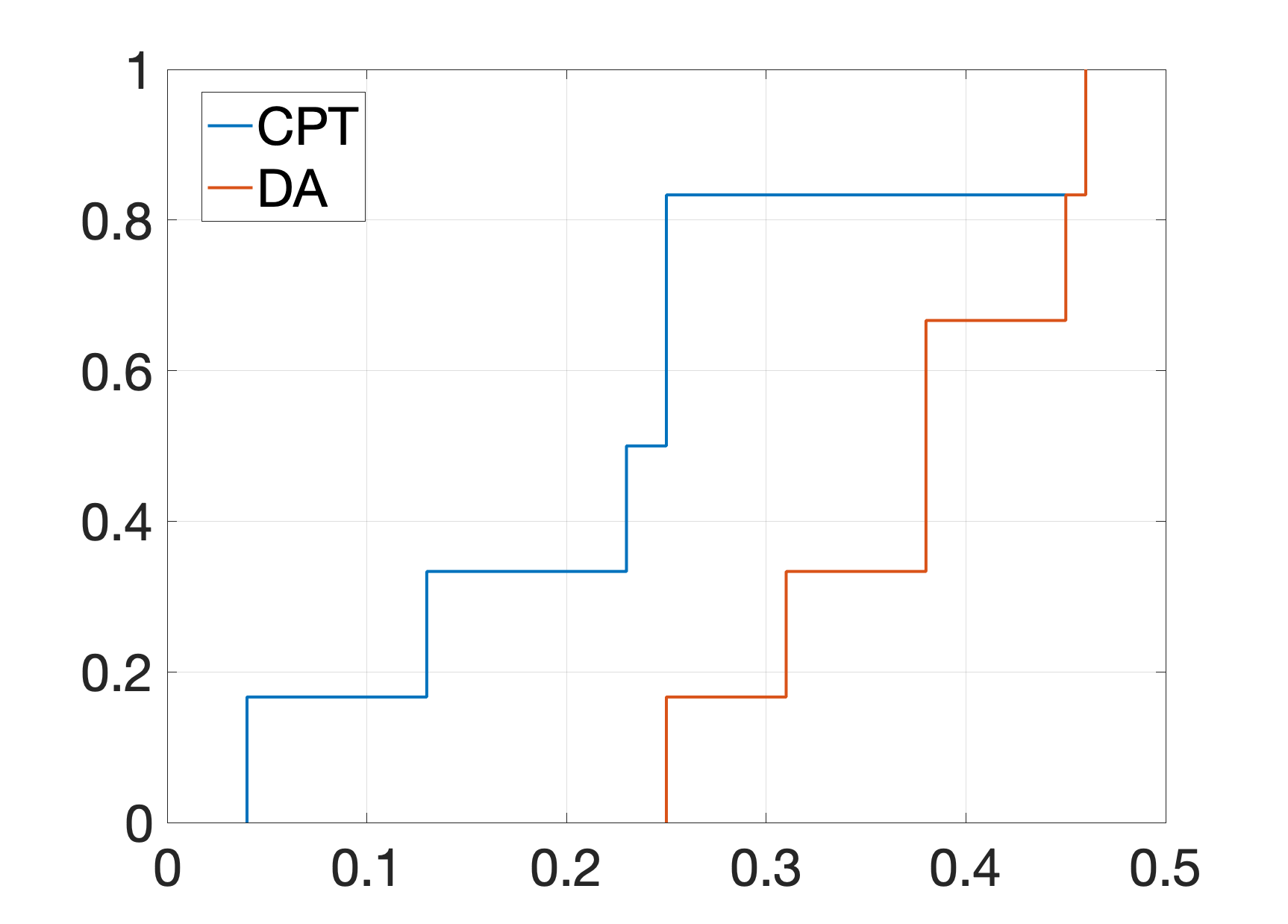}
    \caption{CDF of restrictiveness values}
    \label{fig:CDF}
\end{figure}

\paragraph{Lotteries over the loss domain.} On 25 binary lotteries over the loss domain from \citet{Bruhin}, the 3-parameter specification of CPT indexed to $(\beta,\gamma, \delta)$ predicts the certainty equivalent
$v^{-1}\left((1-w(1-p))\cdot v(\overline{z}) + w(1-p) \cdot v(\underline{z})\right)$ 
for each lottery $(\overline{z},\underline{z},p)$, where
$v(z)= -((-z)^{\beta})$ and
$
w(p)= (\delta p^\gamma)/(\delta p^\gamma + (1-p)^\gamma)
$. The restrictiveness of CPT on these lotteries is 0.31, with a standard error of 0.02.

\paragraph{Lotteries with larger supports.} Finally, we evaluate the restrictiveness of CPT($\alpha,\delta,\gamma$) on gains-domain lotteries with more than two possible outcomes. For each lottery $(z_1,...,z_n; p_1,...p_n)$, where $0\leq z_1 <... < z_n $, the predicted certainty equivalent is
$$v^{-1} \left(  \sum_{i} u(x_i) \left[ w \left( \sum_{k=1}^i p_k \right) - w \left( \sum_{k=1}^{i-1} p_k \right) \right]\right),$$ where for $i=1$ we define $ \sum_{k=1}^{0} p_{k}=0$, and $v$ and $w$ have the same functional forms as used above.  
  On 18 three-outcome gain-domain lotteries from \citet{BernheimSprenger}, the restrictiveness of CPT is 0.57, with a standard error of 0.02.  Thus  CPT is about twice as restrictive for certainty equivalents on  three-outcome lotteries as it is  on binary lotteries.   On a set of 10 six-outcome lotteries from \citet{fudenbergpuri},  the restrictiveness of CPT is $0.83$, with a standard error of 0.01.  These results suggest that CPT is more restrictive on lotteries with larger supports.

\section{\label{sec:App_GameIP}Application 2: The Distribution of Initial Play}\label{games}

\subsection{Setting}
Our second application is to predicting the distribution of initial play in games. Here the feature space $\mathcal{X}$ consists of the 466 unique $3 \times 3$ payoff matrices  from \citet{FudenbergLiang}.\footnote{These data are an aggregate of three data sets: the first is a meta data set of play in 86 games, collected from six experimental game theory papers in \citet{LeytonBrownWright}; the second is a data set of play in 200 games with randomly generated payoffs, which were gathered on MTurk for \citet{FudenbergLiang}; the third is a data set of play in 200 games that were ``algorithmically designed" for a  certain model (level 1 with risk aversion) to perform poorly, again from \citet{FudenbergLiang}.} The outcome space is  the set $\mathcal{Y}=\Delta(\{a_1,a_2,a_3\})$ of distributions of row player actions chosen by the participants in the experiments. The analyst seeks to predict this  distribution for each game.

  For any two prediction rules $f$ and $f'$, we define $d(f,f')$ to be the average Kullback-Liebler divergence between the predicted distributions:
  $d(f,f') = \frac{1}{466} \sum_{x \in \mathcal{X}} D(f(x) \| f'(x))$, 
  where $D$ denotes the Kullback-Leibler divergence. 

We  consider three economic models: The \emph{Poisson Cognitive Hierarchy Model}  (PCHM) of \citet{CamererHoChong04}, the Level-1 model with logistic best replies (henceforth \emph{Logit Level-1}), and the PCHM with logistic best replies (henceforth \emph{Logit PCHM}). The PCHM supposes that there is a distribution over players of differing levels of sophistication: The \emph{level-0} player randomizes uniformly over his available actions,  the \emph{level-1} player best responds to level-0 play \citep{StahlWilson94,StahlWilson95,Nagel}; and for $k\geq 2$, level-$k$ players best respond to a perceived distribution 
\begin{equation}
p_{k}(h,\tau)=\frac{\pi_{\tau}(h)}{\sum_{l=0}^{k-1}\pi_{\tau}(l)}\qquad\forall\,\,h\in\mathbb{N}_{<k} \label{perceiveddistr}
\end{equation}
over (lower) opponent levels, where $\pi_{\tau}$ is the Poisson distribution
with rate parameter $\tau \geq 0$.
The parameter $\tau$ is the single parameter of the model.

The \emph{Logit Level-1} prediction is defined as follows. For each row player action $a_i$, let $\overline{u}(a_{i})$ be the expected payoff of $a_i$ when the column player uses a uniform distribution. The predicted frequency with which $a_i$ is played is
$\exp\left(\lambda \cdot \overline{u}(a_i)\right)/\sum_{i=1}^3 \exp\left(\lambda \cdot \overline{u}(a_i)\right)$, where the logit parameter $\lambda \in \mathbb{R}_+$ is the single parameter of the model.

The \emph{Logit PCHM} (see e.g. \citet{LeytonBrownWright}) replaces the assumption of exact maximization in the PCHM with a logit best response. That is, the level-0 player chooses
$f_0=(1/3,1/3,1/3)$ as in the PCHM, but we recursively construct the distribution of play for higher levels as follows. For each $k\geq 1$, define 
\[v_k(a_i) = \sum_{h=0}^{k-1}     p_k(h,\tau) \left(\sum_{j=1}^3  f_{h}(a_{j}) u(a_i,a_j)\right)\]
to be the expected payoff of action $a_i$ against a player whose type is distributed according to $p_k(\cdot, \tau)$, where $p_k(h,\tau)$ is as given in (\ref{perceiveddistr}). The distribution of play for a level-$k$ player is then 
$f_k(a_i)= \exp(\lambda \cdot v_k(a_i))/\sum_{j=1}^3 \exp(\lambda \cdot v_k(a_j))$, where $\lambda \in \mathbb{R}_+$ is a logit parameter. We  aggregate across levels using a Poisson distribution with rate parameter $\tau \in \mathbb{R}_+$ to yield the predicted distribution of play.

Finally, we define the baseline prediction rule $f_{\text{base}}$ to predict uniform play in every game $x$. This prediction rule is nested in all three models.\footnote{Let $\tau=0$ in the PCHM or Logit PCHM, and let $\lambda=0$ in Logit Level-1.}

\subsection{Completeness}
We evaluate completeness using negative log-loss as the loss function, i.e., if the chosen action is $a_i$ when the model predicts distribution $(p_1, p_2,p_3)$, the loss in that observation is $-\log(p_i)$.\footnote{This loss function is paired to the Kullback-Leibler discrepancy function we used for measuring restrictiveness, see Appendix \ref{sec:Extend} for details.} The models PCHM, Logit Level-1, and Logit PCHM  are 43.6\%, 72.7\%, and 72.9\% complete.  Thus, as observed in a related study by \citet{LeytonBrownWright}, Logit PCHM provides much better predictions of the distribution of play than the baseline PCHM does. Perhaps surprisingly, almost all of Logit PCHM's improved performance can be obtained by simply adding the logit parameter to the Level-1 model; the further improvement from allowing for multiple levels of sophistication is negligible.\footnote{\citet{FudenbergLiang} found  that the Level-1 model provides a good prediction of the  modal action, but  this does not imply that Logit Level-1 will perform  well in predicting the full distribution of play.  The fact that it does further suggests that initial play in many of these experiments is rather unstrategic.}

\subsection{Restrictiveness} \label{sec:RestrictivenessGames}

 We turn now to evaluating the restrictiveness of these models. We have relatively little understanding about their  empirical content, but we do know that they all imply that if an action is strictly dominated, then the frequency with which it is chosen does not exceed 1/3, and that if an action is strictly dominant, then the frequency with which it is chosen is at least $1/3$. We define the eligible set to be all prediction rules that satisfy these conditions.\footnote{In our data, the median frequency of a strictly dominated action is 0.03, and the highest frequency is 0.35; the median frequency for a strictly dominant action is 0.86, and the lowest frequency is 0.69. Payoff maximization implies that dominant strategies should have probability 1 and dominated strategies have probability 0, but this is inconsistent with observed play in most game theory experiments.} 

   All three models are very  restrictive relative to this eligible set: Logit Level-1's restrictiveness is $0.970$, PCHM's restrictiveness is  $0.992$, and Logit PCHM's restrictiveness is 0.971.  Since the models' completeness ranges from 0.436 to 0.729,  they are much better predictors of the real data than of the synthetic data. Table \ref{tab:Games} reports completeness and restrictiveness measures for the models. We find that Logit Level-1 and Logit PCHM are substantially more complete than  PCHM and only slightly less restrictive, but none of the models is dominated by another. Moreover, Logit Level-1  and Logit PCHM are almost identical in terms of completeness and restrictiveness, even though the parametric forms of the two models are not evidently related.\footnote{No value of $\tau$ in the PCHM yields the Level-1 model, so Logit Level-1 is not nested within Logit PCHM.}

\begin{table}
\caption{\label{tab:NetApp}Restrictiveness and Completeness for Initial Play}

\centering{}%
\begin{tabular}{@{}l@{}lcc@{}}
\hline 
 & \# Param & \multicolumn{1}{c}{Restrictiveness} & \multicolumn{1}{c}{Completeness}\tabularnewline
\hline 
\hline 
 PCHM & $\quad$$\quad$ 1 &  0.992 & 0.436 
 \tabularnewline
 & &  ($<$0.001) & (0.017)     \tabularnewline
logit level-1 & $\quad$$\quad$ 1 & 0.970 &   0.727 \tabularnewline
&  & ($<$0.001) & (0.015)  \tabularnewline
logit PCHM & $\quad$$\quad$ 2 & 0.971  & 0.729  \tabularnewline
& & (0.003) & (0.014) \tabularnewline
\hline 
\multicolumn{3}{c}{\footnotesize{Restrictiveness estimated from 1000 simulations.}}
\end{tabular}
\end{table}

Finally, as a robustness check, we consider strengthening the background constraints imposed on the eligible set $\mathcal{F}$.          For each $t\in [0,0.3)$, we define the eligible set $\mathcal{F}(t)$ to include all prediction rules $f$ that satisfy the following conditions: (1) If an action is strictly dominated, then the frequency with which it is chosen does not exceed $1/3-t$; (2) If an action is strictly dominant, then the frequency with which it is chosen is at least $1/3+t$. The constraint imposed by these conditions increases in $t$, and $t=0$ returns our original specification of $\mathcal{F}$. We find that across choices of $t\in [0,0.3)$, the restrictivenesses of PCHM, Logit PCHM, and Logit Level-1 do not fall below 0.89 (see Table \ref{tab:varyFM} below). This tells us that  constraints on the frequency  of strictly dominated and strictly dominant strategies are a very small part of the empirical content of these models.

     \begin{table}[H] \label{tab:Games}
     \centering
           
\begin{tabular}{@{}l@{}ccc@{}}
\hline 
          &  PCHM & Logit Level-1 & Logit PCHM \\
          \hline
          \hline
          max \quad & 0.993 & 0.969 & 0.972\\
          min \quad & 0.974 & 0.890 & 0.957 \\
           \hline
          \end{tabular}
          \caption{\footnotesize{Highest and lowest smallest restrictiveness for $t\in [0,0.3)$.}}
         \label{tab:varyFM}         
          \end{table}

\section{Application 3: Diffusion in Social Networks}

\subsection{Setting}

Our final application is to the prediction of microfinance takeup
rates following diffusion of information in social networks. We use data from a study by \citet{banerjee2013diffusion}, in which certain ``leaders'' in 43 villages  in Karnatka, India were given information about a microfinance program, 
and takeup of the program was then tracked.\footnote{In 2007, the microfinance institution Bharatha Swamukti
Samsthe invited leaders 
within each village to an information
meeting, and asked the leaders to spread the information. The data
set contains the resulting  microfinance takeup rate for each
village and some measures of social connections between households.}

For each village $i$, let $y_{i}$ be the average takeup rate among
non-leader households.\footnote{This is the outcome variable that \citet{banerjee2013diffusion} focus
on.} Our goal is to predict $y_{i}$ given the observed characteristics
$X_{i}$ of village $i$. Specifically, a village configuration $X_{i}:=\left(N_{i},A_{i},L_{i}\right)$
consists of a set $N_{i}$ of villagers, an $n_{i}\times n_{i}$ adjacency
matrix $A_{i}$ that represents the measured social network, and the set $L_{i}$ of leaders in village $i$. The feature space $\mathcal{X}$ is the collection of 43 village configurations, and prediction rules are maps $f:{\cal X}\to\left[0,1\right]$ that from village configurations to the takeup
rate among non-leaders. There are  no obvious a priori restrictions on the takeup rates, so we set $\mathcal{F}$ to be the set $\left[0,1\right]^{43}$ of all possible prediction rules from
${\cal X}$ to $\left[0,1\right]$. We set the discrepancy function as $d(f,g):=\frac{1}{43}\sum_{i=1}^{43} (f(x_i)-g(x_i))^2$ and the loss function as $l(f(x),y) := (f(x)-y)^2.$ 

\subsection{Models}

The first parametric models we consider are OLS regressions with various subsets of the following eight network
statistics as regressors: (1) average eigenvector centrality of leaders; (2) average degree
centrality of leaders; (3) average degree centrality of all villagers;
(4) average betweenness centrality of leaders; (5) clustering coefficient
of village network; (6) average path length in village network; (7)
proportion of connected (non-isolated) villagers; (8) proportion of
leaders.

We compute the restrictiveness and completeness of a sequence of OLS
models by incrementally adding the regressors listed above. We set
the baseline as OLS regression on a constant, which is a special
case of all the linear models we consider. With the loss function  $l(f(x),y):=(y-f(x))^2$, an estimator of completeness (computed based on in-sample errors without the use of cross validations)
reduces to the R squared of the OLS regression.\footnote{Recall that the R-squared of an OLS regression is defined by $R^2:= 1 - SSR/SST$, where $SSR := \sum_i (y_i - x_i'\hat{\b})^2$ is the expected loss under an OLS regression model and $SST := \sum_i (y_i - \ol{y})^2$ is the expected loss under a constant model.} 

We also consider a partially linear
model built upon the ``network gossip centrality'' described
in \citet{banerjee2019using}. To do this, we model each non-leader
household's takeup probability as a function of its position in the village. We define the ``hearing matrix'' of village $i$ by
$H_{i}\left(\t_{1}\right):=\sum_{t=1}^{T}\t_{1}^{t}A_{i}^{t}$,
where $T$ is some given number of time periods for information diffusion.\footnote{$\left(\sum_{t=1}^{T}A_{i}^{t}\right)_{jk}$ counts
the number of paths from $j$ to $k$ of length up to $T$. We set $T = 5$ following \cite{banerjee2019using}.}
With $\t_{1}=1$, the $jk$-th entry of $H_{i}\left(1\right)$ can
be interpreted as the expected number of times villager $k$ hears
a piece of information that originates from villager $j$ within $T$
periods of time. The parameter $\t_{0}\in\left(0,1\right)$ discounts longer paths
of diffusion. For each non-leader $k$ in village $i$, we define
$x_{i,k}\left(\t_{1}\right):=\sum_{j\in L_{i}}\left(H_{i}\left(\t_{1}\right)\right)_{jk}$
as the ``network gossip centrality'' of non-leader $k$, which counts
the (discounted) sum of number of paths from the leaders of village
$i$ to non-leader $k$. Next, we model the takeup probability of
non-leader $k$ as function of $k$'s ``network gossip centrality''
based on a logistic model $p_{i,j}\left(\t_{0},\t_{1}\right):=\frac{\exp\left(\t_{0}+x_{i,j}\left(\t_{1}\right)\right)}{1+\exp\left(\t_{0}+x_{i,j}\left(\t_{1}\right)\right)}$, 
where $\t_{0}$ is a location parameter.\footnote{Note that we do not include a scale parameter here, since if present, it
will be absorbed into $\t_{1}$.} The expected
village-level takeup rate among
non-leaders can then be derived as the average $p_{i,j}(\t_0,\t_1)$ among non-leaders. To allow additional flexibility, and to nest the naive constant model as a special case, we introduce two additional linear parameters $(\t_2,\t_3)$, and set: 
$f_{i}\left(\t\right):=\t_{2}+\t_{3}\cd\frac{1}{\left|N_{i}\backslash L_{i}\right|}\sum_{j\notin L_{i}}p_{ij}\left(\t_{0},\t_{1}\right)$.
This model is very stylized; our purpose is to illustrate how our algorithmic approach can be used to evaluate the restrictiveness of a structural model whose flexibility is otherwise difficult to gauge. 
\subsection{Results}

\begin{table}[h]
\caption{\label{tab:NetApp}Restrictiveness and Completeness for Microfinance Takeup Rates}

\centering{}%
\begin{tabular}{@{}l@{}lcc@{}}
\hline 
 & \# Param & \multicolumn{1}{c}{Restrictiveness} & \multicolumn{1}{c}{Completeness}\tabularnewline
\hline 
\hline 
$\quad\,$Linear Models &  &  & \tabularnewline
\hline 
$\quad\,$Eigenvector Centrality of Leaders & $\quad$$\quad$1 & 0.9762 & 0.2577\tabularnewline
& & (0.0003) & (0.1101)\tabularnewline
+ Degree Centrality of Leaders & $\quad$$\quad$2 & 0.9526 & 0.3385\tabularnewline
& & (0.0004) & (0.1193)\tabularnewline
+ Degree Centrality of All Villagers & $\quad$$\quad$3 & 0.9288 & 0.3471\tabularnewline
& & (0.0005) & (0.1151)\tabularnewline
+ Betweenness Centrality of Leaders & $\quad$$\quad$4 & 0.9053 & 0.3475\tabularnewline
& & (0.0006) & (0.1158)\tabularnewline
+ Clustering Coefficient & $\quad$$\quad$5 & 0.8816 & 0.3516\tabularnewline
& & (0.0007) & (0.1191)\tabularnewline
+ Average Path Length & $\quad$$\quad$6 & 0.8579& 0.3516\tabularnewline
& & (0.0007) & (0.1191)\tabularnewline
+ Proportion of Connected Villagers & $\quad$$\quad$7 & 0.8342 & 0.3575\tabularnewline
& & (0.0008) & (0.1229)\tabularnewline
+ Proportion of Leaders & $\quad$$\quad$8 & 0.8101 & 0.3604\tabularnewline
& & (0.0008) & (0.1237)\tabularnewline
\hline
\hline 
$\quad\,$Partially Linear Model & $\quad$$\quad$4 & 0.9408 & 0.0674\tabularnewline
& & (0.0036) & (0.0452)\tabularnewline
\hline 
\end{tabular}
\end{table}

Table \ref{tab:NetApp} reports the restrictiveness and completeness
of the models described above.\footnote{Table \ref{tab:NetApp} displays the restrictiveness of the linear models based on  M = $10000$ simulations, while restrictiveness for the partially linear models is computed using $M = 100$ simulations. Completeness for all models is computed based the real data with $N = 43$ villages.} The panel ``Linear Models"
contains results about
a sequence of linear models, with a new regressor added to the OLS
regression in each row.\footnote{We add the regressors sequentially according to the ordering above, and omit many other different orderings of the same set of regressors, since the regressions in Table \ref{tab:NetApp} suffice to illustrate our main point.} For example, the row ``+
Degree Centrality'' corresponds to an OLS regression of takeup rates
on a constant, the leaders' average eigenvector, and the leaders' average degree centrality.

The numerical results for linear models are as expected: as more regressors are added the model becomes more
flexible, so restrictiveness decreases while completeness increases.
While restrictiveness seems to be decreasing
at an approximately linear rate starting from the second regression,
the corresponding increases in completeness appear less uniform, and in particular, completeness barely changes when we add the regressor
``average path length in the village.'' Note that this does not
mean that this additional regressor approximately lies in the linear span
of all previously included regressors, since we do observe a nontrivial
reduction in restrictiveness from the addition of this regressor: New regressors eventually barely improve fit to the data, but they continue to decrease restrictiveness. 

A priori it is unclear how restrictive the  partially linear model is.   It turns out that  its restrictiveness is  very high, 0.94, suggesting that the
individual-level modeling of takeup rates as a function
of network gossip centrality imposes substantial restrictions
across village configurations. However, this model's completeness is only 0.07, so  it does not capture much of the variation in village takeup rates.

This four-parameter partially linear model is dominated by the simple
linear model with a constant and the average eigenvector centrality of
leaders as the single regressor: the latter has both  higher restrictiveness
(0.9762 \textgreater{} 0.9408) and higher completeness
(0.2577 \textgreater{} 0.0674).  
This  shows that even a  detailed, structured, and economically-motivated model may turn out to be more flexible than a simple linear model, and that the added flexibility need not help it fit real data.

\section{Conclusion}

When a theory fits the data well, it matters whether  this is  because the theory captures important regularities in the data, or whether the theory is so flexible that it can explain any behavior at all.  We provide a practical, algorithmic approach for evaluating the restrictiveness of a theory, and demonstrate that it reveals new insights into models from two economic domains. The method is easily applied to  models across diverse domains.

As highly flexible machine learning methods become more popular in economics, economic theory is distinguished in part by the structure it imposes on behaviors. We view these restrictions as an important part of the  value added by economic theory, so it is natural to ask how restrictive economic models are compared to the highly flexible approaches used in machine learning. Our restrictiveness  measure offers a way to quantify this.

\bibliographystyle{ecta}
\bibliography{library}

\begin{thebibliography}{51}
\newcommand{\enquote}[1]{``#1''}
\expandafter\ifx\csname natexlab\endcsname\relax\def\natexlab#1{#1}\fi

\bibitem[\protect\citeauthoryear{Abdellaoui, Klibanoff, and Placido}{Abdellaoui
  et~al.}{2015}]{abdellaoui2015experiments}
\textsc{Abdellaoui, M., P.~Klibanoff, and L.~Placido} (2015):
  \enquote{Experiments on compound risk in relation to simple risk and to
  ambiguity,} \emph{Management Science}, 61, 1306--1322.

\bibitem[\protect\citeauthoryear{Andreoni, Gillen, and Harbaugh}{Andreoni
  et~al.}{2013}]{andreoni2013power}
\textsc{Andreoni, J., B.~J. Gillen, and W.~T. Harbaugh} (2013): \enquote{The
  power of revealed preference tests: Ex-post evaluation of experimental
  design,} \emph{Unpublished manuscript}.

\bibitem[\protect\citeauthoryear{Andrews, Fudenberg, Lei, Liang, and
  Wu}{Andrews et~al.}{2022}]{AFLW}
\textsc{Andrews, I., D.~Fudenberg, L.~Lei, A.~Liang, and C.~Wu} (2022):
  \enquote{The Transfer Performance of Economic Models,} Working Paper.

\bibitem[\protect\citeauthoryear{Austern and Zhou}{Austern and
  Zhou}{2020}]{austern2020asymptotics}
\textsc{Austern, M. and W.~Zhou} (2020): \enquote{Asymptotics of
  Cross-Validation,} \emph{arXiv preprint arXiv:2001.11111}.

\bibitem[\protect\citeauthoryear{Ba, Bohren, and Imas}{Ba
  et~al.}{2023}]{ba2023over}
\textsc{Ba, C., J.~A. Bohren, and A.~Imas} (2023): \enquote{Over-and
  Underreaction to Information,} Working Paper.

\bibitem[\protect\citeauthoryear{Banerjee, Chandrasekhar, Duflo, and
  Jackson}{Banerjee et~al.}{2013}]{banerjee2013diffusion}
\textsc{Banerjee, A., A.~G. Chandrasekhar, E.~Duflo, and M.~O. Jackson} (2013):
  \enquote{The diffusion of microfinance,} \emph{Science}, 341.

\bibitem[\protect\citeauthoryear{Banerjee, Chandrasekhar, Duflo, and
  Jackson}{Banerjee et~al.}{2019}]{banerjee2019using}
---\hspace{-.1pt}---\hspace{-.1pt}--- (2019): \enquote{Using gossips to spread
  information: Theory and evidence from two randomized controlled trials,}
  \emph{The Review of Economic Studies}, 86, 2453--2490.

\bibitem[\protect\citeauthoryear{Barberis and Huang}{Barberis and
  Huang}{2008}]{barberis2008stocks}
\textsc{Barberis, N. and M.~Huang} (2008): \enquote{Stocks as lotteries: The
  implications of probability weighting for security prices,} \emph{American
  Economic Review}, 98, 2066--2100.

\bibitem[\protect\citeauthoryear{Barseghyan, Molinari, O’Donoghue, and
  Teitelbaum}{Barseghyan et~al.}{2013}]{Donoghue}
\textsc{Barseghyan, L., F.~Molinari, T.~O’Donoghue, and J.~C. Teitelbaum}
  (2013): \enquote{The Nature of Risk Preferences: Evidence from Insurance
  Choices,} \emph{American Economic Review}, 103, 2499--2529.

\bibitem[\protect\citeauthoryear{Beatty and Crawford}{Beatty and
  Crawford}{2011}]{BeattyCrawford}
\textsc{Beatty, T. and I.~Crawford} (2011): \enquote{How Demanding Is the
  Revealed Preference Approach to Demand?} \emph{American Economic Review},
  101, 2782--95.

\bibitem[\protect\citeauthoryear{Becker}{Becker}{1962}]{Becker}
\textsc{Becker, G.~S.} (1962): \enquote{Irrational behavior and economic
  theory,} \emph{Journal of political economy}, 70, 1--13.

\bibitem[\protect\citeauthoryear{Bernheim and Sprenger}{Bernheim and
  Sprenger}{2020{\natexlab{a}}}]{bernheim2020empirical}
\textsc{Bernheim, B.~D. and C.~Sprenger} (2020{\natexlab{a}}): \enquote{On the
  empirical validity of cumulative prospect theory: Experimental evidence of
  rank-independent probability weighting,} \emph{Econometrica}, 88, 1363--1409.

\bibitem[\protect\citeauthoryear{Bernheim and Sprenger}{Bernheim and
  Sprenger}{2020{\natexlab{b}}}]{BernheimSprenger}
\textsc{Bernheim, D. and C.~Sprenger} (2020{\natexlab{b}}): \enquote{Direct
  Tests of Cumulative Prospect Theory,} Working Paper.

\bibitem[\protect\citeauthoryear{Blow, Browning, and Crawford}{Blow
  et~al.}{2021}]{BlowBrowningCrawford}
\textsc{Blow, L., M.~Browning, and I.~Crawford} (2021):
  \enquote{{Non-parametric Analysis of Time-Inconsistent Preferences},}
  \emph{The Review of Economic Studies}, 88, 2687--2734.

\bibitem[\protect\citeauthoryear{Bronars}{Bronars}{1987}]{Bronars}
\textsc{Bronars, S.} (1987): \enquote{The Power of Nonparametric Tests of
  Preference Maximization,} \emph{Econometrica}, 55, 693--698.

\bibitem[\protect\citeauthoryear{Bruhin, Fehr-Duda, and Epper}{Bruhin
  et~al.}{2010}]{Bruhin}
\textsc{Bruhin, A., H.~Fehr-Duda, and T.~Epper} (2010): \enquote{Risk and
  Rationality: Uncovering Heterogeneity in Probability Distortion,}
  \emph{Econometrica}, 78, 1375--1412.

\bibitem[\protect\citeauthoryear{Camerer, Ho, and Chong}{Camerer
  et~al.}{2004}]{CamererHoChong04}
\textsc{Camerer, C.~F., T.-H. Ho, and J.-K. Chong} (2004): \enquote{A cognitive
  hierarchy model of games,} \emph{The Quarterly Journal of Economics}, 119,
  861--898.

\bibitem[\protect\citeauthoryear{Chen and Santos}{Chen and
  Santos}{2018}]{chen2018overidentification}
\textsc{Chen, X. and A.~Santos} (2018): \enquote{Overidentification in regular
  models,} \emph{Econometrica}, 86, 1771--1817.

\bibitem[\protect\citeauthoryear{Choi, Fisman, Gale, and Kariv}{Choi
  et~al.}{2007}]{Choietal}
\textsc{Choi, S., R.~Fisman, D.~Gale, and S.~Kariv} (2007):
  \enquote{Consistency and Heterogeneity of Individual Behavior under
  Uncertainty,} \emph{American Economic Review}, 97, 1--15.

\bibitem[\protect\citeauthoryear{de~Clippel and Rozen}{de~Clippel and
  Rozen}{2022}]{ClippelRozen}
\textsc{de~Clippel, G. and K.~Rozen} (2022): \enquote{Which Performs Best?
  Comparing Discrete Choice Models,} Working Paper.

\bibitem[\protect\citeauthoryear{Ellis, Kariv, and Ozbay}{Ellis
  et~al.}{2022}]{EllisKarizOzbay}
\textsc{Ellis, K., S.~Kariv, and E.~Ozbay} (2022): \enquote{What Can the Demand
  Analyst Learn from Machine Learning?} Working Paper.

\bibitem[\protect\citeauthoryear{Fan, Budescu, and Diecidue}{Fan
  et~al.}{2019}]{fan2019decisions}
\textsc{Fan, Y., D.~V. Budescu, and E.~Diecidue} (2019): \enquote{Decisions
  with compound lotteries.} \emph{Decision}, 6, 109.

\bibitem[\protect\citeauthoryear{Fehr-Duda and Epper}{Fehr-Duda and
  Epper}{2012}]{fehrdudaepper}
\textsc{Fehr-Duda, H. and T.~Epper} (2012): \enquote{Probability and Risk:
  Foundations and Economic Implication of Probability-Dependent Risk
  Preferences,} \emph{Annual Review of Economics}, 4, 567--593.

\bibitem[\protect\citeauthoryear{Frankel and Kamenica}{Frankel and
  Kamenica}{2019}]{frankel2019quantifying}
\textsc{Frankel, A. and E.~Kamenica} (2019): \enquote{Quantifying information
  and uncertainty,} \emph{American Economic Review}, 109, 3650--80.

\bibitem[\protect\citeauthoryear{Fudenberg, Kleinberg, Liang, and
  Mullainathan}{Fudenberg et~al.}{2022}]{FKLM}
\textsc{Fudenberg, D., J.~Kleinberg, A.~Liang, and S.~Mullainathan} (2022):
  \enquote{Measuring the Completeness of Economic Models,} \emph{Journal of
  Political Economy}, 130, 956--990.

\bibitem[\protect\citeauthoryear{Fudenberg and Liang}{Fudenberg and
  Liang}{2019}]{FudenbergLiang}
\textsc{Fudenberg, D. and A.~Liang} (2019): \enquote{Predicting and
  Understanding Initial Play,} \emph{American Economic Review}, 109,
  4112--4141.

\bibitem[\protect\citeauthoryear{Fudenberg and Puri}{Fudenberg and
  Puri}{2021}]{fudenbergpuri}
\textsc{Fudenberg, D. and I.~Puri} (2021): \enquote{Evaluating and Extending
  Theories of Choice Under Risk,} Working Paper.

\bibitem[\protect\citeauthoryear{Goldstein and Einhorn}{Goldstein and
  Einhorn}{1987}]{goldstein1987expression}
\textsc{Goldstein, W.~M. and H.~J. Einhorn} (1987): \enquote{Expression theory
  and the preference reversal phenomena,} \emph{Psychological review}, 94,
  236--254.

\bibitem[\protect\citeauthoryear{Green and Hwang}{Green and
  Hwang}{2012}]{green2012initial}
\textsc{Green, T.~C. and B.-H. Hwang} (2012): \enquote{Initial public offerings
  as lotteries: Skewness preference and first-day returns,} \emph{Management
  Science}, 58, 432--444.

\bibitem[\protect\citeauthoryear{Gul}{Gul}{1991}]{gul1991}
\textsc{Gul, F.} (1991): \enquote{A Theory of Disappointment Aversion,}
  \emph{Econometrica}, 59, 667--686.

\bibitem[\protect\citeauthoryear{Hansen}{Hansen}{1982}]{hansen1982large}
\textsc{Hansen, L.~P.} (1982): \enquote{Large sample properties of generalized
  method of moments estimators,} \emph{Econometrica}, 50, 1029--1054.

\bibitem[\protect\citeauthoryear{Harless and Camerer}{Harless and
  Camerer}{1994}]{HarlessCamerer}
\textsc{Harless, D. and C.~Camerer} (1994): \enquote{The Predictive Utility of
  Generalized Expected Utility Theories,} \emph{Econometrica}, 62, 1251--1289.

\bibitem[\protect\citeauthoryear{Hausman}{Hausman}{1978}]{hausman1978specification}
\textsc{Hausman, J.~A.} (1978): \enquote{Specification tests in econometrics,}
  \emph{Econometrica}, 46, 1251--1271.

\bibitem[\protect\citeauthoryear{Hey}{Hey}{1998}]{Hey}
\textsc{Hey, J.~D.} (1998): \enquote{An application of Selten’s measure of
  predictive success,} \emph{Mathematical Social Sciences}, 35, 1--15.

\bibitem[\protect\citeauthoryear{Karmarkar}{Karmarkar}{1978}]{karmarkar1979}
\textsc{Karmarkar, U.} (1978): \enquote{Subjectively weighted utility: A
  descriptive extension of the expected utility model,} \emph{Organizational
  Behavior \& Human Performance}, 21, 67--72.

\bibitem[\protect\citeauthoryear{Koopmans and Reiersol}{Koopmans and
  Reiersol}{1950}]{Koopmans}
\textsc{Koopmans, T. and O.~Reiersol} (1950): \enquote{The Identification of
  Structural Characteristics,} \emph{The Annals of Mathematical Statistics},
  21, 165--181.

\bibitem[\protect\citeauthoryear{Lattimore, Baker, and Witte}{Lattimore
  et~al.}{1992}]{lattimore1992influence}
\textsc{Lattimore, P.~K., J.~R. Baker, and A.~D. Witte} (1992): \enquote{The
  influence of probability on risky choice: A parametric examination,}
  \emph{Journal of Economic Behavior \& Organization}, 17, 315--436.

\bibitem[\protect\citeauthoryear{Murad, Sefton, and Starmer}{Murad
  et~al.}{2016}]{murad2016risk}
\textsc{Murad, Z., M.~Sefton, and C.~Starmer} (2016): \enquote{How do risk
  attitudes affect measured confidence?} \emph{Journal of Risk and
  Uncertainty}, 52, 21--46.

\bibitem[\protect\citeauthoryear{Nagel}{Nagel}{1995}]{Nagel}
\textsc{Nagel, R.} (1995): \enquote{Unraveling in Guessing Games: An
  Experimental Study,} \emph{American Economic Review}, 85, 1313--1326.

\bibitem[\protect\citeauthoryear{Peysakhovich and Naecker}{Peysakhovich and
  Naecker}{2017}]{PeysakhovichNaecker}
\textsc{Peysakhovich, A. and J.~Naecker} (2017): \enquote{Using methods from
  machine learning to evaluate behavioral models of choice under risk and
  ambiguity,} \emph{Journal of Economic Behavior and Organization}, 133,
  373--384.

\bibitem[\protect\citeauthoryear{Polisson, Quah, and Renou}{Polisson
  et~al.}{2020}]{Polisson}
\textsc{Polisson, M., J.~K.-H. Quah, and L.~Renou} (2020): \enquote{Revealed
  Preferences over Risk and Uncertainty,} \emph{American Economic Review}, 110,
  1782--1820.

\bibitem[\protect\citeauthoryear{Routledge and Zin}{Routledge and
  Zin}{2010}]{routledge2010generalized}
\textsc{Routledge, B.~R. and S.~E. Zin} (2010): \enquote{Generalized
  disappointment aversion and asset prices,} \emph{The Journal of Finance}, 65,
  1303--1332.

\bibitem[\protect\citeauthoryear{Sargan}{Sargan}{1958}]{sargan1958estimation}
\textsc{Sargan, J.~D.} (1958): \enquote{The estimation of economic
  relationships using instrumental variables,} \emph{Econometrica}, 26,
  393--415.

\bibitem[\protect\citeauthoryear{Schwaninger}{Schwaninger}{2022}]{Schwaninger}
\textsc{Schwaninger, M.} (2022): \enquote{Sharing with the powerless third:
  Other-regarding preferences in dynamic bargaining,} \emph{Journal of Economic
  Behavior and Organization}, 197, 341--355.

\bibitem[\protect\citeauthoryear{Selten}{Selten}{1991}]{Selten}
\textsc{Selten, R.} (1991): \enquote{Properties for a Measure of Predictive
  Success,} \emph{Mathematical Social Sciences}, 21, 153--167.

\bibitem[\protect\citeauthoryear{Snowberg and Wolfers}{Snowberg and
  Wolfers}{2010}]{SnowbergWolfers}
\textsc{Snowberg, E. and J.~Wolfers} (2010): \enquote{Explaining the
  Favorite-Long Shot Bias: Is It Risk-Love or Misperceptions?} \emph{Journal of
  Political Economy}, 118, 723--746.

\bibitem[\protect\citeauthoryear{Stahl and Wilson}{Stahl and
  Wilson}{1994}]{StahlWilson94}
\textsc{Stahl, D.~O. and P.~W. Wilson} (1994): \enquote{Experimental evidence
  on players' models of other players,} \emph{Journal of Economic Behavior and
  Organization}, 25, 309--327.

\bibitem[\protect\citeauthoryear{Stahl and Wilson}{Stahl and
  Wilson}{1995}]{StahlWilson95}
---\hspace{-.1pt}---\hspace{-.1pt}--- (1995): \enquote{On players' models of
  other players: Theory and experimental evidence,} \emph{Games and Economic
  Behavior}, 10, 218--254.

\bibitem[\protect\citeauthoryear{Sutter, Kocher, Gl{\"a}tzle-R{\"u}tzler, and
  Trautmann}{Sutter et~al.}{2013}]{sutter2013impatience}
\textsc{Sutter, M., M.~G. Kocher, D.~Gl{\"a}tzle-R{\"u}tzler, and S.~T.
  Trautmann} (2013): \enquote{Impatience and uncertainty: Experimental
  decisions predict adolescents' field behavior,} \emph{American Economic
  Review}, 103, 510--31.

\bibitem[\protect\citeauthoryear{Tversky and Kahneman}{Tversky and
  Kahneman}{1992}]{CPT}
\textsc{Tversky, A. and D.~Kahneman} (1992): \enquote{Advances in Prospect
  Theory: Cumulative Representation of Uncertainty,} \emph{Journal of Risk and
  Uncertainty}, 5, 297--323.

\bibitem[\protect\citeauthoryear{Wright and Leyton-Brown}{Wright and
  Leyton-Brown}{2014}]{LeytonBrownWright}
\textsc{Wright, J.~R. and K.~Leyton-Brown} (2014): \enquote{Level-0 meta-models
  for predicting human behavior in games,} \emph{Proceedings of the fifteenth
  ACM conference on Economics and computation}, 857--874.

\end{thebibliography}

\appendix



\section{Proof of Proposition \ref{prop:axiom}} \label{Proof_Axioms}

Throughout this proof, we use  $\Sigma$ to denote the Lebesgue $\sigma$-algebra on $\overline{\mathcal{F}}$, and shorten $\Sigma$-measurable to simply ``measurable."

It is clear that A\ref{axiom:Nonnegative}-A\ref{axiom:Convexity} are satisfied by the representation in (\ref{eq:e}), and A\ref{axiom:Nonnegative}-A\ref{axiom:Symmetry} are satisfied by the approximation error measure given in (\ref{e:Uniform}).
For the other direction, we begin by demonstrating the following lemma:

\begin{lemma} \label{lemm:Expectation} Suppose $e$ satisfies  A\ref{axiom:Nonnegative} and A\ref{axiom:Convexity}. Then for every $\mathcal{F}_\Theta$ and $d$, there exists a function $h: \overline{\mathcal{F}} \rightarrow \mathbb{R}$ such that 
$e(\mathcal{F}_\Theta,\mathcal{F},d) = \mathbb{E}\left[h(f) \, : \, f \sim \mu_{\mathcal{F}}\right] $ for all   measurable sets $\mathcal{F}$.
\end{lemma}

\begin{proof} Fix an arbitrary $\mathcal{F}_\Theta$ and $d$, and define $e_*: \Sigma \rightarrow \mathbb{R}$  to satisfy
$e_*(\mathcal{F}) \equiv e(\mathcal{F}_\Theta,\mathcal{F},d)$ for all measurable $\mathcal{F}$. 
The lemma follows if we can show that A\ref{axiom:Convexity} implies the existence of a function $h: \overline{\mathcal{F}} \rightarrow \mathbb{R}$ such that
$e_*(\mathcal{F}) = \int_{\mathcal{F}} h(f) d\mu_{\mathcal{F}}$ for all measurable $\mathcal{F},$ where $\mu_{\mathcal{F}}$ denotes the  measure $\mu$ conditional on the event $\mathcal{F}$.

Define $\nu: \Sigma \rightarrow \mathbb{R}$ to satisfy
$\nu(\mathcal{F}) = \mu(\mathcal{F}) \cdot e_*(\mathcal{F})$ for all measurable $\mathcal{F}$. Then A\ref{axiom:Convexity} implies that for any sequence $\mathcal{F}_{\Theta_1}, \mathcal{F}_{\Theta_2}, \dots$,
$\sum_{i=1}^\infty \nu(\mathcal{F}_{\Theta_i}) 
 = \nu\left(\bigcup_{i=1}^\infty \mathcal{F}_{\Theta_i}\right)$.
Also, $\nu(\emptyset)=0$ (since $\mu(\emptyset)=0$) and $\nu$ is non-negative (by A\ref{axiom:Nonnegative}), so $\nu$ is a measure on $(\overline{\mathcal{F}},\Sigma)$. 
Moreover, $\nu$ is absolutely continuous with respect to  $\mu$ by construction. So the Radon-Nikdoym theorem implies existence of a function $h: \overline{\mathcal{F}}\rightarrow \mathbb{R}$ such that 
$\nu(\mathcal{F}) = \int_{\mathcal{F}} h(f) d\mu$  for all measurable $\mathcal{F}.$
Then $\mu(\mathcal{F}) e_*(\mathcal{F}) = \mu(\mathcal{F}) \int_{\mathcal{F}} h(f) \frac{d\mu}{\mu(\mathcal{F})}  = \mu(\mathcal{F}) \int_{\mathcal{F}} h(f) d\mu_{\mathcal{F}},$
so $e_*(\mathcal{F}) = \int_{\mathcal{F}} h(f) d\mu_{\mathcal{F}}$.
\end{proof}

Now fix any $\mathcal{F}_\Theta$ and $d$, and let $h$ be the function given in Lemma \ref{lemm:Expectation}. We will show that A\ref{axiom:Monotone} and A\ref{axiom:Rescale} imply that for each $f \in \overline{\mathcal{F}}$,
\begin{equation} \label{eq:h}
h(f) = c_f \cdot d(\mathcal{F}_\Theta,f)
\end{equation}
for some constant $c_f \in \mathbb{R}_+$.

Fix an arbitrary $f$. Lemma \ref{lemm:Expectation} implies 
$e(\mathcal{F}_\Theta,\{f\},d)  = \int h(f') \cdot d\delta_{f} = h(f),$ 
where $\delta_f$ denotes the Dirac measure at $f$. So it is sufficient for (\ref{eq:h}) to show that there is a constant $c_f \in \mathbb{R}_+$ such that 
$e(\mathcal{F}_\Theta,\{f\},d) = c_f \cdot d(\mathcal{F}_\Theta,f)$  for all  $\mathcal{F}_\Theta,d.$
By A\ref{axiom:Monotone},  models can be completely ordered for the eligible set $\{f\}$, where $e(\mathcal{F}_{\Theta_1},\{f\},d) \geq e(\mathcal{F}_{\Theta_2},\{f\},d)$ if and only if $d(\mathcal{F}_{\Theta_1},f) \geq d(\mathcal{F}_{\Theta_2},f)$. So there is a monotone increasing function $\Phi: \mathbb{R} \rightarrow \mathbb{R}$ such that
\begin{equation} \label{eq:increasing}
e(\mathcal{F}_\Theta,\{f\},d) = \Phi(d(\mathcal{F}_\Theta,f)).
\end{equation}

Now we will show that $\Phi$ must be linear. Choose an arbitrary $\alpha \in \mathbb{R}_+$. Define $d' = \alpha \cdot d$ and suppose some model $\mathcal{F}_{\Theta'}$ satisfies $d(\mathcal{F}_{\Theta'},f) = \alpha \cdot d(\mathcal{F}_\Theta,f)$. Then 
$e(\mathcal{F}_\Theta,\{f\},d')  = \alpha \cdot e(\mathcal{F}_\Theta,\{f\},d) = \alpha \cdot \Phi \left(d(\mathcal{F}_\Theta,f)\right),$ 
where the first equality follows  by (A\ref{axiom:Rescale}) and the second follows by (\ref{eq:increasing}). Also  
$e(\mathcal{F}_{\Theta'},\{f\},d) = \Phi ( d(\mathcal{F}_{\Theta'}, f))  = \Phi (\alpha \cdot d(\mathcal{F}_\Theta,f)),$ where the first equality follows by (\ref{eq:increasing}). A\ref{axiom:Rescale} requires $e(\mathcal{F}_{\Theta'},\{f\},d)  = e(\mathcal{F}_\Theta,\{f\}, d')$, so $\alpha \cdot \Phi (d(\mathcal{F}_\Theta,f)) = \Phi(\alpha \cdot d(\mathcal{F}_\Theta,f))$. Thus we can write
$e(\mathcal{F}_\Theta,\{f\},d) = c_{f} \cdot d(\mathcal{F}_\Theta,f)$
for some constant $c_{f} \in \mathbb{R}_+$. Repeating this argument for every $f$, there is a function $c: \mathcal{F} \rightarrow \mathbb{R}$ such that
$e(\mathcal{F},\mathcal{F},d) = \mathbb{E}_{f \sim \mu_{\mathcal{F}}}\left[ c(f) \cdot d(G,f)  \right]$ for all measurable $\mathcal{F}$,
so we have the representation in (\ref{eq:e}). 

Now suppose that A\ref{axiom:Symmetry} is satisfied in addition to the other axioms. The previous arguments imply that there is a function $c: \overline{\mathcal{F}} \rightarrow \mathbb{R}$ such that 
\[e(\mathcal{F}_\Theta,\mathcal{F},d) =   \mathbb{E}_{f \sim \mu_{\mathcal{F}}}\left[c(f) \cdot \inf_{f' \in \mathcal{F}_\Theta} d(f',f)  \right] \quad \forall \mathcal{F}_\Theta,\mathcal{F},d \]
Suppose towards contradiction that $e$ cannot be represented by (\ref{e:Uniform}). Then there must exist an eligible set $\mathcal{F}$ and  $f,f' \in \mathcal{F}$ such that $c(f) \cdot \mu_{\mathcal{F}}(f) > c(f') \cdot \mu_{\mathcal{F}}(f')$. But then for any models $\mathcal{F}_{\Theta_1}$ and $\mathcal{F}_{\Theta_2}$ with the property that $[d(\mathcal{F}_{\Theta_1},f) = d(\mathcal{F}_{\Theta_2},f') > d(\mathcal{F}_{\Theta_2},f) = d(\mathcal{F}_{\Theta_1},f'),$ it follows that $e(\mathcal{F}_{\Theta_1},\{f,f'\},d) > e(\mathcal{F}_{\Theta_2},\{f,f'\},d)$, violating A\ref{axiom:Symmetry}. 

\clearpage
\noindent \begin{center}
{\large{}Online Appendix to the Paper}
\par\end{center}{\large \par}
\noindent \begin{center}
{\LARGE{}How Flexible is that Functional Form? Measuring the Restrictiveness of Theories}
\par\end{center}{\LARGE \par}
\medskip{}
\begin{center}
Drew Fudenberg \quad Wayne Gao \quad Annie Liang
\par\end{center}{\large \par}
\noindent \begin{center}
\today
\par\end{center}


\section{A Guide for Practitioners} \label{guide}

Below we provide detailed instructions for how to take the proposed measures to other applications.

\subsection{Setup}

\paragraph{The Prediction Problem and Model.} We suppose that the researcher has a dataset that can be described as a set of observations $(x,y)$, where $x$ is interpreted as an observable input, and $y$ is interpreted as the outcome to be predicted. Define
\begin{itemize}
    \item the \textbf{set of features} $\mathcal{X}$ to consist of all unique instances of $x$ in the analyst's data (thus by construction finite).
    \item the \textbf{set of outcomes} $\mathcal{Y} \subseteq \mathbb{R}^k$ to be the set in which $y$ takes values. 
\end{itemize}
Let $\ol{\cF} = \mathcal{Y}^{\vert\mathcal{X}\vert}$ be the set of all mappings from $\mathcal{X}$ to $\mathcal{Y}$. 

The researcher is interested in studying the properties of some parametric model $\cF_\T=\{f_\theta\}_{\theta\in \Theta}$, where each  $f_\theta$ belongs to $\ol{\cF}$.

\paragraph{Baseline.}Choose a ``baseline mapping" $f_{\text{base}}$ from the model $\cF_\T$. The purpose of the baseline is to provide a lower bound for error that any sensible model should outperform. Some possibilities for how to choose this baseline include:
\begin{itemize}
    \item choosing a ``degenerate" version of the model with the parameters fixed at some default values (for example, Expected Value as  a degenerate case of Cumulative Prospect Theory, as in our Application 1)
    \item choosing a mapping that corresponds to ``guessing at random" (e.g., predicting a uniform distribution over the possible outcomes, as in our Application 2)
    \item choosing a best constant prediction based on the data (e.g., regressing a linear model on a constant, as in our Application 3)
\end{itemize}

The choice of baseline mapping should be reported along with estimates of restrictiveness and completeness, and a natural robustness check is to verify that these estimates do not change significantly over different (reasonable) choices of baseline.

\subsection{Evaluating Restrictiveness}

\paragraph{The Eligible Set.} The researcher first determines the \textbf{eligible set} $\cF$, which is a subset of mappings from $\mathcal{X}$ to $\mathcal{Y}$ that satisfy some given properties.  Which and how many properties to choose depends on what the researcher wants to understand. If the researcher wants to know whether the model imposes any restrictions at all, then  the eligible set should include all mappings from $\mathcal{X}$ to $\mathcal{Y}$. If the researcher wants to know how restrictive the model is beyond imposing some Property A, then the eligible set should include only mappings that are consistent with Property A.

\paragraph{The Discrepancy Function $d$.} Next the researcher chooses a discrepancy function $d: \ol{\cF} \times \ol{\cF} \rightarrow \mathbb{R}_+$ that tells us how different any two mappings $f$ and $f'$ are. Although we leave this specification open to the researcher, we recommend choice of a continuous $d$ to facilitate computation. Additionally, when the outcome space $\mathcal{Y}$ is real-valued, a natural choice is the expected squared distance between the predictions of $f$ and $f'$, namely \[d(f,f')=\mathbb{E}_{P_X}\left[(f(X)-f'(X))^2\right]\]
where $P_X$ is the empirical distribution on $\mathcal{X}$ in the researcher's dataset. And when the outcome space $\mathcal{Y}$ consists of probability distributions, a natural choice is the expected Kullback-Liebler divergence between  the predictions of $f$ and $f'$, namely \[d(f,f')=\mathbb{E}_{P_X}\left[D(f(X) \| f'(X))\right]\]
where $D$ denotes the Kullback-Liebler divergence. Nonstandard choices of $d$ should be explained and justified.

\paragraph{Computing Restrictiveness.} By assumption that $\mathcal{Y}$ is a subset of finite-dimensional Euclidean space, the uniform distribution on any choice of eligible set $\mathcal{F}$ is well-defined. To compute the restrictiveness $r(\cF_\T, \cF)$ for a parametric model $\cF_\T$, the researcher should:
\begin{enumerate}
    \item Choose a sample size $M \in \N$ (for example, set $M = 1000$).
    \item Sample $M$ mappings from the uniform distribution on the eligible set $\mathcal{F}$. Denote each generated mapping by $f_m$.
    \item Compute the estimate of restrictiveness as follows:
    \[
    \hat r = \frac{\frac{1}{M} \sum_{m=1}^{M} d(\cF_\T, f_m) }{\frac{1}{M} \sum_{m=1}^{M} d(f_{base}, f_m)}.
    \]
    where $d(\cF_\T,f) \equiv \inf_{g \in \cF_\T} d(g,f)$.
\end{enumerate}

When $d$ is continuous (as is recommended), then $d(\cF_\T,f) \equiv \inf_{g \in \cF_\T} d(g,f)$ can be replaced by $d(\cF_\T,f) \equiv \min_{g \in \cF_\T} d(g,f)$, which can be computed for example by discretizing $\Theta$ and searching over this grid.

\paragraph{Computing the Standard Error.} Let $\hat \sigma^2_{\cF_{\Theta}}$ be the sample variance of $\{d(\cF_{\Theta}, f_m)\}_{m=1}^M$, $\hat{\sigma}_{\{f_{base}\}}$ be the sample variance of $\{d(\{f_{base}\}, f_m)\}_{m=1}^M$, and $\hat \sigma_{\cF_{\Theta}, \{f_{base}\}}$ be the sample covariance of $\{d(\cF_{\Theta}, f_m)\}_{m=1}^M$ and $\{d(\{f_{base}\}, f_m)\}_{m=1}^M$. Define
\[
\hat \sigma^2_{\hat r} \equiv 
\frac{\hat \sigma^2_{\cF_\T} - 2 \cdot \hat r \cdot \hat \sigma_{\cF_\T, \{f_{base}\}} + \hat{r}_M^2 \cdot \hat \sigma^2_{\{f_{base}\}}}{\left(\frac{1}{M} \sum_{m=1}^M d(f_{base}, f_m) \right)^2}
\]
Then, $\sqrt M (\hat r - r(\cF_\T, \ol{\cF}))/\hat \sigma_{\hat r} \dto \mathcal{N}(0,1)$, so the standard error of $\hat{r}$ can be estimated by $\hat \sigma_{\hat r}/\sqrt{M}$.

\subsection{Evaluating Completeness}

\paragraph{The Loss Function $\ell$.} Choose a continuous loss function $\ell: \ol{\cF} \times \mathcal{X} \times \mathcal{Y} \rightarrow \mathbb{R}_+$ where $\ell(f,(x,y))$ measures how wrong the prediction $f(x)$ is when the true outcome is $y$. We leave this specification open to the researcher, but there are natural choices of loss functions to use depending on the prediction problem and the choice of discrepancy $d$. As we discuss in Appendix \ref{sec:Extend}, certain choices of discrepancy $d$ and loss $\ell$ are ``paired," and thus are natural to choose with one another. Specifically, when the outcome space $\mathcal{Y}$ is real-valued and the discrepancy $d$ is the expected squared distance, then consider choosing
\[\ell(f,x,y) = (y-f(x))^2\]
to be the squared distance between the prediction and the outcome. When the outcome space $\mathcal{Y}$ consists of a set of probability distributions and the discrepancy $d$ is the expected KL divergence, then consider choosing
\[\ell(f,(x,y)) = -\log f(y\mid x)\]
to be the negative conditional log-likelihood of observing $y$ given $x$. 

\paragraph{Computing Completeness.} Let the researcher's data be written as $\left\{ Z_{i}:=\left(X_{i},Y_{i}\right)\right\} _{i=1}^{N}$. 
We describe below a $K$-fold cross-validated estimator for completeness $\kappa(\cF_\T)$. 

For $\widetilde{\mathcal{F}} \in \{\ol{\cF}, \cF_\T, \{f_{\text{base}}\}\}$, compute the respective out-of-sample prediction errors  $\hat{e}_{CV}\left(\ol{\cF} \right)$, $\hat{e}_{CV}\left(\cF_\T\right)$ and $\hat{e}_{CV}\left(f_{\text{base}}\right)$  as follows: 
\begin{enumerate}
    \item Divide the data $(Z_1, \dots, Z_N)$ into $K$ (approximately)
equal-sized groups. To simplify notation, assume that $J_{N}=\frac{N}{K}$
is an integer.
\item Let $k\left(i\right)$ denote the group number of observation
$Z_{i}$. In each $k$-th iteration of cross-validation, the $k$-th test set consists of all observations belonging to group $k$, and the $k$-th training set consists of all remaining observations. 
\item For each group $k=1,...,K$, define 
\[
\hat{f}^{-k} :=\arg\min_{f\in \widetilde{{\cal F}}}\frac{1}{N-J_{N}}\sum_{k\left(i\right)\neq k}l(f,Z_i)\]
to be the element of  $\widetilde{\cal F}$  that minimizes error for prediction of the training data in iteration $k$.  This estimated mapping is used for prediction of the $k$-th test set, and 
\[
\hat{e}_{k}  :=\frac{1}{J_{N}}\sum_{k\left(i\right)=k}l\left(\hat{f}^{-k},Z_i\right)
\]
is its out-of-sample error. 
\item Then,
\[\hat{e}_{CV}\left(\widetilde{{\cal F}}\right)  :=\frac{1}{K}\sum_{k=1}^{K}\hat{e}_{k}  
\]
is the average out-of-sample error across the $K$ choices of test set.
\end{enumerate}

The following is an estimator for $\kappa(\cF_\T)$:
\[
\hat{\kappa}= 1 - \frac{\hat{e}_{CV}\left(\cF_\T\right)-\hat{e}_{CV}\left({\ol{\cF}}\right)}{\hat{e}_{CV}\left(f_{\text{base}}\right)-\hat{e}_{CV}\left({\ol{\cF}}\right)}.
\]

\paragraph{Computing the Standard Error.} For the $k$-th test set, let $f_{\hat{\t}^{-k}}$ and $\hat{f}^{-k}$
be the estimated mappings from models $\cF_\T$ and $\cF$,
respectively. The difference in their test errors on observation $Z_{i}$
is $$\Delta_{\t,k}\left(Z_{i}\right):=l\left(f_{\hat{\t}^{-k}},Z_{i}\right)-l\left(\hat{f}^{-k},Z_{i}\right),$$
and the average difference across all observations in test fold $k$
is 
\[
\overline{\Delta}_{\t,k}:=\frac{1}{J_{N}}\sum_{k(i)=k}\Delta_{\t,k}\left(Z_{i}\right).
\]
The sample variance of the difference in test errors for the $k$-th fold is 
\[
\hat{\s}_{\Delta_{\t},k}^{2}:=\frac{1}{J_{N}-1}\sum_{k(i)=k}\left(\Delta_{\t,k}\left(Z_{i}\right)-\overline{\Delta}_{\t,k}\right)^{2}
\]
which we then average over the $K$ folds and obtain
\[
\hat{\s}_{\Delta_{\t}}^{2}:=\frac{1}{K}\sum_{k=1}^{K}\hat{\s}_{\D_{\t},k}^{2}.
\]
Similarly we define $\Delta_{f_{\text{base}},k}\left(Z_{i}\right):=l\left(f_{\text{base}},Z_{i}\right)-l\left(\hat{f}^{-k},Z_{i}\right)$, and correspondingly $\overline{\Delta}_{f_{\text{base}},k}$, $\hat{\s}_{\Delta_{f_{\text{base}}},k}^{2}$ and 
$\hat{\s}_{\Delta_{f_{\text{base}}}}^{2}$. Lastly, define
the covariance estimator by 
\[
\hat{\s}_{\Delta_{\text{\ensuremath{\t}}}\D_{f_{\text{base}}}}:=\frac{1}{K}\sum_{k=1}^{K}\frac{1}{J_{N}-1}\sum_{k(i)=k}\left(\Delta_{\t,k}\left(Z_{i}\right)-\overline{\Delta}_{\t,k}\right)\left(\Delta_{f_{\text{base}},k}\left(Z_{i}\right)-\ol{\Delta}_{f_{\text{base}},k}\left(Z_{i}\right)\right).
\]
Based on $\hat{\s}_{\Delta_{\t}}^{2},\hat{\s}_{\Delta_{f_{\text{base}}}}^{2}$
and $\hat{\s}_{\Delta_{\text{\ensuremath{\t}}}\D_{f_{\text{base}}}}$,
we define the following variance estimator for $\hat{\kappa}$:
\begin{equation}
\hat{\s}_{\hat{\kappa}}^{2}:=\frac{\hat{\s}_{\Delta_{\t}}^{2}-2\hat{\kappa}\hat{\s}_{\Delta_{\text{\ensuremath{\t}}}\D_{f_{\text{base}}}}+\hat{\kappa}^{2}\hat{\s}_{\Delta_{f_{\text{base}}}}^{2}}{\left[\hat{e}_{CV}\left(f_{\text{base}}\right)-\hat{e}_{CV}\left({\ol{\cF}}\right)\right]^{2}}\label{eq:sekappa}
\end{equation}
so the standard error of $\hat{\kappa}$ can be estimated by $\hat{\s}_{\hat{\kappa}}/\sqrt{N}$.

\section{Proof of Proposition \ref{ANorm_kstar}} \label{sec:asym_CV}

\subsection{Preliminary Definitions}

We now introduce some definitions and notation that will be useful in the derivation of the asymptotic distribution of the CV-based completeness estimator. 

\subsubsection{Finite-Sample Out-of-Sample Error}

Let ${\bf Z}_{N}:=\left(Z_{i}\right)_{i=1}^{N}$ be a random sample of observations in a given data set, and let $Z_{N+1}\sim P$ denote a random variable with the same distribution $P$ that is independent of ${\bf Z}_{N}$. For a given data set ${\bf Z}_{N}$ and a given model $\tilde{\cal F}$, we define the conditional out-of-sample error (given data set ${\bf Z}_{N}$) as 
\[
e_{\tilde{\cal F}}\left({\bf Z}_{N}\right):=\E\left[\rest{l\left(\hat{f}_{{\bf Z}_{N}}, Z_{N+1}\right)}{\bf Z}_{N}\right],
\]
where $\hat{f}_{{\bf Z}_{N}}\in\tilde{\cal F}$ is an estimator, or an algorithm, that selects a mapping $\hat{f}_{{\bf Z}_{N}}$ within the model $\tilde{\cal F}$ based on data ${\bf Z}_{N}$. We also define the out-of-sample error, with expectation taken over different possible data sets ${\bf Z}_{N}$, as 
$
e_{\tilde{\cal F},N}:=\E\left[e_{\tilde{\cal F}}\left({\bf Z}_{N}\right)\right].
$

From the definition of the K-fold cross-validation estimator, it can be shown that $\E\left[\hat{e}_{CV}\left(\tilde{\cal F}\right)\right]=e_{{\cal F},\frac{K-1}{K}N}$. The asymptotic distribution of $\hat{e}_{CV}\left(\tilde{\cal F}\right)-e_{\tilde{\cal F},\frac{K-1}{K}N}$
has been studied in the statistics and machine learning literature. Our analysis below will be based on the results in \citet{austern2020asymptotics} on the asymptotic distribution of $\hat{e}_{CV}\left(\tilde{\cal F}\right)-e_{\tilde{\cal F},\frac{K-1}{K}N}$.

\subsubsection{Joint Parametrization of $\cF_\T$ and $\ol{\cF}$}

Recall that the model $\cF_\T$ is parametrized by $\theta \in \Theta$, and $f_\theta$ denotes a generic function in $\cF_\T$. Since $\cX$ is finite,  $\ol{\cF}$ can be parameterized by a finite-dimensional parameter $\b\in{\cal B}\subseteq\R^{d_{\ol{\cF}}}$ and use the notation $f_{\left[\b\right]}\in\ol{\cF}$ to denote a generic function in $\ol{\cF}$. Since by assumption $f^* \in \ol{\cF}$, we can define a parameter $\b^{*}$ to represent it, i.e. $f_{\left[\b^{*}\right]}=f^{*}$.

For arbitrary $\theta$ and $\beta$, write 
$
l_{\T}\left(\t, Z_{i}\right):= l \left( f_\t, Z_i\right)$ and $l_{\cal B} \left(\b, Z_{i}\right):= l\left(f_{[\b]}, Z_i\right).
$
We define the estimation mappings by
$\hat{\theta}\left({\bf Z}_{N}\right) :=\arg\min_{\theta\in\Theta}\frac{1}{N}\sum l_{\T}\left(\t, Z_i\right)$ and $
\hat{\b}\left({\bf Z}_{N}\right) :=\arg\min_{\b\in {\cal B}_\cM}\frac{1}{N}\sum l_{{\cal B}}\left(\b, Z_i\right).$ Let $\a:=\left(\t^{'},\b^{'}\right)^{'}$ denote the concatenation of the parameters $\t \in \T$ and $\b \in \cal{B}$,
$\a^{*}:=\left(\t^{*'},\b^{*'}\right)^{'}$ to be the parameters associated with the best mappings in $\cF_\T$ and $\ol{\cF}$, and also define
\[\hat{\a}\left({\bf Z}_{N}\right) :=\left(\hat{\t}^{'}\left({\bf Z}_{N}\right),\hat{\b}^{'}\left({\bf Z}_{N}\right)\right)^{'} =\arg\min_{\theta\in\Theta,\b\in{\cal B}}\frac{1}{N}\sum_{i=1}^{N}\left[l_{\T}\left(\t, Z_i\right)+l_{{\cal {\cal B}}}\left(\b, Z_i\right)\right]\]
to be an estimator for $\a^*$. Finally, define
\begin{align*}
\D l\left(\t,\b; Z_i\right) & :=l\left(f_{\t}, Z_i\right)-l\left(f_{[\b]}, Z_i\right)=l_{\T}\left(\t, Z_i\right)-l_{{\cal B}}\left(\b, Z_i\right).
 \end{align*}

\subsection{Construction of Variance Estimator}\label{app:VarEst}

To obtain the standard error of the estimator, we use a variance estimator adapted
from Proposition 1 in \citet{austern2020asymptotics}. Specifically,
for the $k$-th test set, let $f_{\hat{\t}^{-k}}$ and $\hat{f}^{-k}$
be the estimated mappings from models $\cF_\T$ and $\ol{\cF}$,
respectively. The difference in their test errors on observation $Z_{i}$
is $\Delta_{\t,k}\left(Z_{i}\right):=l\left(f_{\hat{\t}^{-k}},Z_{i}\right)-l\left(\hat{f}^{-k},Z_{i}\right),$
and the average difference across all observations in test fold $k$
is 
$
\overline{\Delta}_{\t,k}:=\frac{1}{J_{N}}\sum_{k(i)=k}\Delta_k\left(Z_{i}\right).
$
The sample variance of the difference in test errors for the $k$-th fold is 
\[
\hat{\s}_{\Delta_{\t},k}^{2}:=\frac{1}{J_{N}-1}\sum_{k(i)=k}\left(\Delta_{\t,k}\left(Z_{i}\right)-\overline{\Delta}_{\t,k}\right)^{2}
\]
which we  average over the $K$ folds and obtain
$
\hat{\s}_{\Delta_{\t}}^{2}:=\frac{1}{K}\sum_{k=1}^{K}\hat{\s}_{\D_{\t},k}^{2}.$

Similarly we define $\Delta_{f_{\text{base}},k}\left(Z_{i}\right):=l\left(f_{\text{base}},Z_{i}\right)-l\left(\hat{f}^{-k},Z_{i}\right)$, and correspondingly $\overline{\Delta}_{f_{\text{base}},k}$, $\hat{\s}_{\Delta_{f_{\text{base}}},k}^{2}$ and 
$\hat{\s}_{\Delta_{f_{\text{base}}}}^{2}$. Lastly, define
the covariance estimator by 
\[
\hat{\s}_{\Delta_{\text{\ensuremath{\t}}}\D_{f_{\text{base}}}}:=\frac{1}{K}\sum_{k=1}^{K}\frac{1}{J_{N}-1}\sum_{k(i)=k}\left(\Delta_{\t,k}\left(Z_{i}\right)-\overline{\Delta}_{\t,k}\right)\left(\Delta_{f_{\text{base}},k}\left(Z_{i}\right)-\ol{\Delta}_{f_{\text{base}},k}\left(Z_{i}\right)\right).
\]
Based on $\hat{\s}_{\Delta_{\t}}^{2},\hat{\s}_{\Delta_{f_{\text{base}}}}^{2}$
and $\hat{\s}_{\Delta_{\text{\ensuremath{\t}}}\D_{f_{\text{base}}}}$,
we define the following variance estimator for $\hat{\kappa}$:
\begin{equation}
\hat{\s}_{\hat{\kappa}}^{2}:=\frac{\hat{\s}_{\Delta_{\t}}^{2}-2\hat{\kappa}\hat{\s}_{\Delta_{\text{\ensuremath{\t}}}\D_{f_{\text{base}}}}+\hat{\kappa}^{2}\hat{\s}_{\Delta_{f_{\text{base}}}}^{2}}{\left[\hat{e}_{CV}\left(f_{\text{base}}\right)-\hat{e}_{CV}\left({\ol{\cF}}\right)\right]^{2}}.\label{eq:sekappa}
\end{equation}

\subsection{Material Based on \citet{austern2020asymptotics}}
\begin{assumption}[Conditions for Asymptotics of CV Estimator]
 \label{assu:CVasymp} \phantom{1}

\begin{enumerate}
\item $l_{\T}\left(\t, z\right)$ and $l_{{\cal B}}\left(\b, z\right)$ are
twice differentiable and strictly convex in $\t$ and $\b$.
\item $\E\left[\sup_{\t\in\T}l_{\T}^{4}\left(\t, Z_i\right)\right]<\infty$
and $\E\left[\sup_{\b\in{\cal B}}l_{{\cal B}}^{4}\left(\b, Z_i\right)\right]<\infty$. 
\item There exist open neighborhoods ${\cal O}_{\t^{*}}$ and ${\cal O}_{\b^{*}}$
of $\t^{*}$and $\b^{*}$ in $\T$ and ${\cal B}$ such that 

\begin{enumerate}
\item $\E\left[\sup_{\t\in{\cal O}_{\t^{*}}}\norm{\Dif_{\t}l_{\T}\left(\t,Z_i\right)}^{16}\right]<\infty$, $\E\left[\sup_{\b\in{\cal O}_{\b^{*}}}\norm{\Dif_{\b}l_{{\cal B}}\left(\b,Z_i\right)}^{16}\right]<\infty.$ 
\item $\E\left[\sup_{\t\in{\cal O}_{\t^{*}}}\norm{\Dif_{\t}^{2}l_{\T}\left(\t,Z_i\right)}^{16}\right]<\infty$,
$\E\left[\sup_{\b\in{\cal O}_{\b^{*}}}\norm{\Dif_{\b}l_{{\cal B}}\left(\b,Z_i\right)}^{16}\right]<\infty.$ 
\item there exists $c>0$ such that $\l_{min}\left(\Dif_{\t}^{2}l_{\T}\left(\t,Z_i\right)\right)\geq c$, $\l_{min}\left(\Dif_{\b}^{2}l_{{\cal B}}\left(\b,Z_i\right)\right)\geq c$ a.s.
uniformly on ${\cal O}_{\t^{*}}$ and ${\cal O}_{\b^{*}}$.
\end{enumerate}
\end{enumerate}
\end{assumption}

\begin{lemma}
\label{lem:asym_norm_CV}Under Assumption \ref{assu:CVasymp}: 
\[
\sqrt{N}\left[\hat{e}_{CV}\left(\cF_\T\right)-\hat{e}_{CV}\left(\ol{\cF}\right)-\left(e_{\cF_\T,\frac{K-1}{K}N}-e_{\ol{\cF},\frac{K-1}{K}N}\right)\right]\dto\cN\left(0,\textup{Var}\left(\D l\left(f_{\t^{*}},f^{*};Z_i\right)\right)\right).\]
\end{lemma}

\begin{proof}
Proposition 5 of \citet{austern2020asymptotics} establishes the asymptotic normality of cross-validation risk estimator and its asymptotic variance under parametric settings where the loss function used for training is the same as the loss function used for evaluation. Applying Proposition 5 of \citet{austern2020asymptotics} under Assumption
\ref{assu:CVasymp} to $\t,\b$ and $\a=\left(\t,\b\right)$, we obtain:
\begin{align*}
\sqrt{N}\left(\hat{e}_{CV}\left(\cF_\T\right)-e_{\cF_\T,\frac{K-1}{K}N}\right) & \dto\cN\left(0,\textup{Var}\left(l\left(f_{\t^{*}},Z_i\right)\right)\right),\\
\sqrt{N}\left(\hat{e}_{CV}\left(\ol{\cF}\right)-e_{\cF,\frac{K-1}{K}N}\right) & \dto\cN\left(0,\textup{Var}\left(l\left(f^{*},Z_i\right)\right)\right),\\
\sqrt{N}\left(\hat{e}_{CV}\left(\cF_\T\right)+\hat{e}_{CV}\left(\ol{\cF}\right)-e_{\cF_\T,\frac{K-1}{K}N}-e_{\ol{\cF},\frac{K-1}{K}N}\right) & \dto\cN\left(0,\textup{Var}\left(l\left(f_{\t^{*}},Z_i\right)+l\left(f^{*},Z_i\right)\right)\right).
\end{align*}
Using the equality $\textup{Var}\left(X+Y\right)+\textup{Var}\left(X-Y\right)=2\textup{Var}\left(X\right)+2\textup{Var}\left(Y\right)$,
we then deduce that
\[
\sqrt{N}\left[\hat{e}_{CV}\left(\cF_\T\right)-\hat{e}_{CV}\left(\ol{\cF}\right)-\left(e_{\cF_\T,\frac{K-1}{K}N}-e_{\ol{\cF},\frac{K-1}{K}N}\right)\right]\dto\cN\left(0,\textup{Var}\left(\D l\left(f_{\t^{*}},f^{*};Z_{i}\right)\right)\right).
\]
\end{proof}

\begin{lemma}[Application of Proposition 1 of \citealp{austern2020asymptotics}]
 \label{lem:asym_var_CV}Under Assumption \ref{assu:CVasymp}, 
$\hat{\s}_{\D}^{2}\pto \textup{Var}\left(\D l\left(f_{\t^{*}},f^{*};Z_{i}\right)\right).$
\end{lemma}

\begin{proof}
Applying Proposition 1 of \citet{austern2020asymptotics} under Assumption \ref{assu:CVasymp} to $\t,\b$ and $\a=\left(\t,\b\right)$:
\begin{align*}
\hat{\s}_{\cF_\T}^{2} & :=\frac{1}{K}\sum_{k=1}^{K}\frac{1}{J_{N}-1}\sum_{k\left(i\right)=k}\left(l\left(f_{\hat{\t}^{-k}},Z_{i}\right)-\frac{1}{J_{N}}\sum_{k\left(j\right)=k}l\left(f_{\hat{\t}^{-k}},Z_{j}\right)\right)^{2}
  \pto \textup{Var}\left(l\left(f_{\t^{*}},Z_i\right)\right),
\end{align*}
\begin{align*}
\hat{\s}_{\ol{\cF}}^{2} & :=\frac{1}{K}\sum_{k=1}^{K}\frac{1}{J_{N}-1}\sum_{k\left(i\right)=k}\left(l\left(f_{\left[\hat{\b}^{-k}\right]},Z_{i}\right)-\frac{1}{J_{N}}\sum_{k\left(j\right)=k}l\left(f_{\left[\hat{\b}^{-k}\right]},Z_{j}\right)\right)^{2}
  \pto \textup{Var}\left(l\left(f^{*},Z_i\right)\right),
\end{align*}
and 
\begin{align*}
 & \hat{\s}_{\cF_\T+\ol{\cF}}^{2}\\
:=\  & \frac{1}{K}\sum_{k=1}^{K}\frac{1}{J_{N}-1}  \sum_{k\left(i\right)=k}\left(l\left(f_{\hat{\t}^{-k}},Z_i\right)+l\left(f_{\left[\hat{\b}^{-k}\right]},Z_{i}\right)-\frac{1}{J_{N}}\sum_{k\left(j\right)=k}\left[l\left(f_{\left[\hat{\b}^{-k}\right]},Z_{j}\right)+l\left(f_{\hat{\t}^{-k}},Z_{i}\right)\right]\right)^{2}\\
\pto\  & \textup{Var}\left(l\left(f_{\t^{*}},Z_{i}\right)+l\left(f^{*},Z_{i}\right)\right).
\end{align*}

Hence,
$\hat{\s}_{\D_\t}^{2}  =2\hat{\s}_{\cF_\T}^{2}+2\hat{\s}_{\ol{\cF}}^{2}-\hat{\s}_{\cF_\T+\ol{\cF}}^{2}
 \pto 2\textup{Var}\left(l\left(f_{\t^{*}},Z_{i}\right)\right)+2\textup{Var}\left(l\left(f^{*},Z_{i}\right)\right)-2\textup{Var}\left(l\left(f_{\t^{*},Z_{i}}\right)+l\left(f^{*},Z_{i}\right)\right) =\textup{Var}\left(\D l\left(f_{\t^{*}},f^{*};Z_{i}\right)\right).$
\end{proof}

\subsection{Finishing the Proof}

Lemma \ref{lem:asym_norm_CV} characterizes the limit distribution of 
\[
\sqrt{N}\left[\hat{e}_{CV}\left(\cF_\T\right)-\hat{e}_{CV}\left(\ol{\cF}\right)-\left(e_{\cF_\T,\frac{K-1}{K}N}-e_{\ol{\cF},\frac{K-1}{K}N}\right)\right]
\]
which we show is also the limit distribution of $\sqrt{N}\left[\hat{e}_{CV}\left(\cF_\T\right)-\hat{e}_{CV}\left(\ol{\cF}\right)-\left(e_{\cF_\T}-e_{\ol{\cF}}\right)\right].$

To see this, notice that 
\begin{align*}
e_{\cF_\T,\frac{K-1}{K}N}-e_{\cF_\T} &= \E\left[l_{\T}\left(\hat{\t}^{-k\left(i\right)},Z_{i}\right)-l_{\T}\left(\t^{*},Z_{i}\right)\right]\\
=\  & \E\left[\Dif l_{\T}\left(\t^{*},Z_{i}\right)\cd\left(\hat{\t}^{-k\left(i\right)}-\t^{*}\right)+\left(\hat{\t}^{-k\left(i\right)}-\t^{*}\right)^{'}\Dif^{2}l_{\T}\left(\tilde{\t},Z_{i}\right)\cd\left(\hat{\t}^{-k\left(i\right)}-\t^{*}\right)\right]\\
=\  & 0+\E\left[\left(\hat{\t}^{-k\left(i\right)}-\t^{*}\right)^{'}\Dif^{2}l_{\T}\left(\tilde{\t},Z_{i}\right)\cd\left(\hat{\t}^{-k\left(i\right)}-\t^{*}\right)\right]\\
=\  & \frac{1}{N-J_{N}}\E\left[\sqrt{N-J_{N}}\left(\hat{\t}^{-k\left(i\right)}-\t^{*}\right)^{'}\Dif^{2}l_{\T}\left(\tilde{\t},Z_{i}\right)\cd\sqrt{N-J_{N}}\left(\hat{\t}^{-k\left(i\right)}-\t^{*}\right)\right]\\
=\  & c\frac{1}{N-J_{N}}+o\left(\frac{1}{N-J_{N}}\right)=c\frac{K}{K-1}\cd\frac{1}{N}+o\left(\frac{1}{N}\right)
\end{align*}
since $J_{N}=N/K$. Therefore
$\sqrt{N}\left(e_{\T,\frac{K-1}{K}N}-e_{\T}\right)=o_{p}\left(1\right),$
and $\sqrt{N}\left(e_{\ol{\cF},\frac{K-1}{K}N}-e_{\ol{\cF}}\right)=o_{p}\left(1\right)$. Hence: 
$
\sqrt{N}\left[\hat{e}_{CV}\left(\cF_\T\right)-\hat{e}_{CV}\left(\ol{\cF}\right)-\left(e_{\cF_\T}-e_{\ol{\cF}}\right)\right]\dto\cN\left(0,\textup{Var}\left(\D l\left(f_{\t^{*}},f^{*};Z_{i}\right)\right)\right).
$

Now, we replicate the previous result with $f_{\text{base}}$ in place of $\cF_\T$ and obtain
\begin{align*}
\sqrt{N}\left[\hat{e}_{CV}\left(f_{\text{base}}\right)-\hat{e}_{CV}\left(\ol{\cF}\right)-\left(e_{f_{\text{base}}}-e_{\ol{\cF}}\right)\right] & \dto\cN\left(0,\textup{Var}\left(\D l\left(f_{\text{base}},f^{*};Z_{i}\right)\right)\right).
\end{align*}
and jointly
\[
\sqrt{N}\left(\begin{array}{c}
\hat{e}_{CV}\left(\cF_\T\right)-\hat{e}_{CV}\left(\ol{\cF}\right)-\left(e_{\cF_\T}-e_{\ol{\cF}}\right)\\
\hat{e}_{CV}\left(f_{\text{base}}\right)-\hat{e}_{CV}\left(\ol{\cF}\right)-\left(e_{f_{\text{base}}}-e_{\ol{\cF}}\right)
\end{array}\right)\dto\cN\left({\bf 0},\left(\begin{array}{cc}
\s_{\D_{\t}}^{2} & \s_{\D_{\t}\D_{f_{\text{base}}}}\\
\s_{\D_{\t}\D_{f_{\text{base}}}} & \s_{\D_{f_{\text{base}}}}^{2}
\end{array}\right)\right)
\]
with
$\s_{\D_{\t}}^{2} :=\textup{Var}\left(\D l\left(f_{\t^{*}},f^{*};Z_{i}\right)\right)$, $\s_{\D_{f_{\text{base}}}}^{2}  :=\textup{Var}\left(\D l\left(f_{\text{base}},f^{*};Z_{i}\right)\right)$, and 
$\s_{\D_{\t}\D_{f_{\text{base}}}}  :=\textup{Cov}\left(\D l\left(f_{\t^{*}},f^{*};Z_{i}\right),\D l\left(f_{\text{base}},f^{*};Z_{i}\right)\right).$

By Lemma \ref{lem:asym_var_CV}, Assumption \ref{assu:naive} and
the Delta Method, we have 
\[
\sqrt{N}\left(\hat{\kappa}-\kappa\right)\dto \cN\left(0,\ \frac{\s_{\D_{\t}}^{2}-2\kappa\s_{\D_{\text{\ensuremath{\t}}}\D_{f_{\text{base}}}}+\kappa^{*2}\s_{\D_{f_{\text{base}}}}^{2}}{d^{2}\left(f_{\text{base}},f^{*}\right)}\right).
\]
Since $\hat{\s}_{\hat{\kappa}} \pto \left(\s_{\D_{\t}}^{2}-2\kappa\s_{\D_{\text{\ensuremath{\t}}}\D_{f_{\text{base}}}}+\kappa^{*2}\s_{\D_{f_{\text{base}}}}^{2}\right)/d^{2}\left(f_{\text{base}},f^{*}\right)$,
we have $\sqrt{N}\left(\hat{\kappa}-\kappa\right)/\hat{\s}_{\hat{\kappa}}\dto\cN\left(0,1\right)$.

\section{Supplementary Material to Application 1}
\subsection{Estimates for Application 1} \label{app:table}

\begin{table}[H]
\caption{Restrictiveness and Completeness for Certainty Equivalents}

\centering{}%
\begin{tabular}{@{}c@{}lcc@{}}
\hline 
 & \# Param & \multicolumn{1}{c}{Restrictiveness} & \multicolumn{1}{c}{Completeness}\tabularnewline
\hline 
\hline 
CPT Specifications &  &  & \tabularnewline
\hline 
$\alpha,\delta,\gamma$ & $\quad$$\quad$ 3 &  0.28   & 0.95 
 \tabularnewline
 & &  (0.003) & (0.02)     \tabularnewline
 
$\delta,\gamma$ & $\quad$$\quad$ 2 &  0.37   & 0.95
 \tabularnewline
 & & (0.004)  & (0.02) 
 \tabularnewline
 
$\alpha,\gamma$ & $\quad$$\quad$ 2 &  0.51   & 0.95
 \tabularnewline
 & & (0.006)  & (0.02) 
 \tabularnewline
 
 $\alpha,\delta$ & $\quad$$\quad$ 2 &   0.49  & 0.27
 \tabularnewline
 & &  (0.005)  & (0.05)
 \tabularnewline
 
 $\alpha$ & $\quad$$\quad$ 1 &  0.91   & 0.25
 \tabularnewline
 & & (0.005) & (0.05) 
 \tabularnewline

 $\delta$ & $\quad$$\quad$ 1 &   0.68  & 0.26
 \tabularnewline
 & & (0.009) & (0.06)
 \tabularnewline

 $\gamma$ & $\quad$$\quad$ 1 &   0.59  & 0.71
 \tabularnewline
 & & (0.006) & (0.06)
 \tabularnewline
 
 \hline
 DA Specifications &  &  & \tabularnewline
\hline 
$\alpha,\eta$ & $\quad$$\quad$ 2 &  0.47  & 0.27  
 \tabularnewline
 & & (0.006) & (0.06)
 \tabularnewline
 
 $\eta$ & $\quad$$\quad$ 1 &  0.69    & 0.27
 \tabularnewline
 & & (0.009) & (0.05)
 \tabularnewline
\hline
\end{tabular}
      \label{tab:gains}
\end{table}

    Restrictiveness is estimated from 1000 simulations and we report the analytic standard errors. Because of potential dependence among the reported certainty equivalents of subjects, we compute the standard errors for completeness using a block bootstrapping procedure that clusters together all observations from the same subject.\footnote{When generating a bootstrap sample, we randomly sample the 179 subjects with replacement, and include all the reported certainty equivalents of the drawn subjects with replacement.} We then carry out our (cross-validated) estimation of completeness on each bootstrap sample, and compute the standard errors based on 1000 bootstrap samples. These bootstrapped standard errors are similar to the analytic standard errors we get under a revision of the formulas in Section \ref{Est} to accommodate clustering on subjects (see the following section).
    
\subsection{Analytical SE with Clustering} \label{subsec:analyticalcluster}

We discuss here an alternative method for calculating clustered standard errors for completeness.

We consider each subject's reported certainty equivalents for the 25 lotteries as a 25-dimensional vector. We assume that this 25-dimensional vector is i.i.d. across subjects, but leave the dependence within this subject-specific vector unrestricted.  Specifically, define the feature space $\mathcal{X}$ to be a singleton consisting of the $25 \times 3$ matrix whose rows are the different lottery tuples $(\overline{z},\underline{z},p)$ in the \citet{Bruhin} data. The outcome space is $\mathcal{Y}=\mathbb{R}^{25}$, where a typical element is a vector of 25 certainty equivalents for the 25 lotteries. The expected certainty equivalent vector over subjects is represented by a mapping $f: \mathcal{X} \rightarrow \mathbb{R}^{25}$, which is simply a vector in $\mathbb{R}^{25}$.

Finally, let the loss function $l$ be  
$$l(f,Y_i,X) := \frac{1}{25}\norm{Y_i-f_\t(X)}^2 = \frac{1}{25}\sum_{h=1}^{25}(Y_{i,h}-f_{h})^2.$$
This loss function groups together the squared losses of each individual subject across the 25 lotteries. Under this setup,  the analytical formula for standard errors provided in Section $\ref{sec:estimateComplete}$ and Appendix $\ref{app:VarEst}$ can be directly applied, with sample size $N = 179$. Table \ref{tab:analytic} reports the standard errors for completeness computed in this way.

\begin{table}[H]
\centering{}%
\begin{tabular}{@{}c@{}lcc@{}}
\hline 
 & \# Param &  \multicolumn{1}{c}{Completeness}\tabularnewline
\hline 
\hline 
CPT Specifications &    & \tabularnewline
\hline 
$\alpha,\delta,\gamma$ & $\quad$$\quad$ 3  & 0.95 
 \tabularnewline
 &  & (0.09)     \tabularnewline
 
$\delta,\gamma$ & $\quad$$\quad$ 2    & 0.95
 \tabularnewline
 &   & (0.08) 
 \tabularnewline
 
$\alpha,\gamma$ & $\quad$$\quad$ 2 &    0.95
 \tabularnewline
 &  & (0.09) 
 \tabularnewline
 
 $\alpha,\delta$ & $\quad$$\quad$ 2  & 0.27
 \tabularnewline
 &   & (0.09)
 \tabularnewline
 
 $\alpha$ & $\quad$$\quad$ 1   & 0.25
 \tabularnewline
 &  & (0.05) 
 \tabularnewline

 $\delta$ & $\quad$$\quad$ 1   & 0.26
 \tabularnewline
 & & (0.06)
 \tabularnewline

 $\gamma$ & $\quad$$\quad$ 1   & 0.71
 \tabularnewline
 & & (0.06)
 \tabularnewline
 
 \hline
 DA Specifications &    & \tabularnewline
\hline 
$\alpha,\eta$ & $\quad$$\quad$ 2   & 0.27  
 \tabularnewline
 &  & (0.06)
 \tabularnewline
 
 $\eta$ & $\quad$$\quad$ 1   & 0.27
 \tabularnewline
 &  & (0.05)
 \tabularnewline
\hline

\end{tabular}
      \label{tab:analytic}

\end{table}
    \subsection{Restrictiveness on Alternative Sets of Lotteries} \label{sec:TableROther}
    
We report here the restrictiveness values used to construct the CDFs in Figure \ref{fig:CDF} as well as the papers the corresponding sets of lotteries were derived from, and the number of lotteries  from each paper. 

\begin{table}[H]
\caption{Restrictiveness}

\centering{}%
\begin{tabular}{cccc}

\hline 

            Source Paper & \# Lotteries & CPT$(\alpha, \delta, \gamma)$ & DA$(\alpha, \eta)$  \\\hline \hline
             \citet{abdellaoui2015experiments} & 3 &0.04 & 0.31\\
            & &(0.00)& (0.01) \\
            \citet{murad2016risk} & 25 &0.25 & 0.38 \\
            &&(0.00) & (0.00)\\
              \citet{sutter2013impatience} & 4 &0.46& 0.46 \\
            &&(0.01)& (0.01)\\
             \citet{fan2019decisions} & 19 &0.23& 0.25 \\
            &&(0.00)& (0.00)\\
             \citet{bernheim2020empirical} & 7 &0.13& 0.45\\
            &&(0.00)& (0.01)\\
            \hline
        \end{tabular}
\end{table}

\section{``Pairing" Completeness and Restrictiveness}  \label{sec:Extend}

In this section, we show that completeness and restrictiveness are related via the equation
\begin{equation} \label{eq:relationship}
\kappa(\cF_\T) = 1 - r(\cF_\T,\ol{\cF}),
\end{equation}
when the loss function $l$ used to define $e_{P}$, and the discrepancy function $d$ used to define $r$, are ``paired" in a coherent way, which we now explain.

We first provide more details about the formulation of completeness. Suppose that besides $X$, there is a random outcome $Z$. We will consider hypothetical joint distributions $\widetilde{P}$ with different conditional distribution $\widetilde{P}_{Z|X}$, where the marginal distribution $\widetilde{P}_X$ is held fixed. The analyst wants to learn a statistic of the conditional distribution of $Z$ given $X$, which we denote by 
$Y\in {\cal {Y}}$. Two leading cases of this problem are: (a) prediction of the conditional expectation $\mathbb{E}_{\widetilde{P}}[\rest Z X]$, and (b) prediction of the conditional distribution $\widetilde{P}_{Z\mid X}$ itself. 
 As in the main text, a prediction is any function $f: \mathcal{X} \rightarrow \mathcal{Y}$, and we define $\ol{\cF}$ to be the set of all such mappings.
 
Let $l: \ol{\cF} \times \mathcal{X} \times \mathcal{Z} \rightarrow \mathbb{R}$ be a loss function, where  $l(f,(x,z))$ is the loss assigned to predicting $f(x)$ when the realized outcome is $z$. We define the expected error of a prediction rule $f$ with respect to the distribution $\widetilde{P}$ by
\begin{equation} \label{eq:error}
e_{\widetilde{P}}\left(f\right)  :=\E_{\widetilde{P}}\left[l(f,(X,Z))\right],
\end{equation}
and let $f_{\widetilde{P}}^*$ denote the prediction rule that minimizes the expected error under $\widetilde{P}$:
\[f_{\widetilde{P}}^* := \min_{f \in \mathcal{F}} e_{\widetilde{P}}(f).\]
As in the main text, $P$ denotes the distribution from which real data is generated. Then the completeness of a model $\cF_\T$ as defined in \citet{FKLM} can be written as
\[\kappa(\cF_\T) = \frac{e_{P}(f_{\text{base}}) - e_{P}(\cF_\T)}{e_{P}(f_{\text{base}}) - e_{P}(\ol{\cF})} \equiv 1 - \frac{e_{P}(\cF_\T) - e_{P}(\ol{\cF})}{e_{P}(f_{\text{base}}) - e_{P}(\ol{\cF})}.\]

We now formally define the meaning of ``pairing" between the discrepancy function $d$ and the loss function $l$.

\begin{definition}
The loss function $l$ and discrepancy $d: \ol{\cF} \times \ol{\cF} \rightarrow \mathbb{R}$ are \emph{paired} if 
\begin{equation} \label{eq:decompose}
d(f,f_{\widetilde{P}}^*) = e_{\widetilde{P}}(f) - e_{\widetilde{P}}(f_{\widetilde{P}}^*)
\end{equation}
for every distribution $\widetilde{P} \in \Delta(\mathcal{X} \times \mathcal{Z})$ whose marginal distribution on $\mathcal{X}$ is $P_X$. That is, $d(f,f_{\widetilde{P}}^*)$ is the difference between the error of prediction rule $f$ and the error of the best prediction rule $f_{\widetilde{P}}^*$.\footnote{This relation resembles but differs from  the  coupling of the ``cost of uncertainty'' and the ``value of information'' in  \cite{frankel2019quantifying}, which concerns comparisons of different signal structures, as opposed to comparing model classes.}  
\end{definition}

As noted in the main text, if $l$ and $d$ are paired, then (\ref{eq:relationship}) holds, where $f^* = f_{P}^*$. Moreover, as also noted in the main text, the following functions are  paired:
\begin{itemize}
    \item Let $\mathcal{Y} = \mathbb{R}$. Then squared loss $
l\left(f,(x,z)\right):=\left(z-f(x)\right)^{2}$ and the squared distance discrepancy
$
d_{MSE}(f,g):=\E_{P_{X}}\left[\left(f\left(X\right)-g\left(X\right)\right)^{2}\right]$ 
are paired.
\item  Let $\mathcal{Y}$ be the set of distributions over a finite set $\mathcal{Z}$.
Then  negative (conditional) log-likelihood
$
l\left(f,(x,z)\right):=-\log f\left(\rest zx\right)
$  and the KL-divergence discrepancy
\[
d_{KL}(f,g):=\E_{P_{X}}\left[\sum_{z \in \mathcal{Z}}g\left(\rest zx\right)\left[\log g\left(\rest zx\right)-\log f\left(\rest zx\right)\right]\right]\]
are paired.
\end{itemize}

\smallskip

\subsection{A Loss Function That Cannot be Paired with any Discrepancy} \label{app:LAD}

\noindent When $\mathcal{Y}$ is the set of distributions on $\mathcal{Z}$, then every loss function $l$ has a paired discrepancy function, since we can define $d(f,f_{\widetilde{P}}) := e_{f_{\widetilde{P}}}(f) - e_{f_{\widetilde{P}}}(f_{\widetilde{P}})$.\footnote{This is because  $\widetilde{P}$ is completely pinned down by $f_{\widetilde{P}}$ given $P_X$, so $e_{\widetilde{P}} = e_{f_{\widetilde{P}}}$.} But in general, for some prediction problems and loss functions $l$, there may not exist a discrepancy $d$ such that $l$ and $d$ are paired, as the next example shows. In these cases, we can still evaluate restrictiveness and completeness, but they will not have an evident relationship.

Consider a setting where $X$ is degenerate, i.e., ${\cal X}$ is a singleton, so that the joint distribution $\widetilde{P}$ is completely characterized by the distribution of $Y$. Furthermore, let ${\cal Y}:=\left[0,1\right]$. If 
$
f^{*}:=\text{med}\left(Y\right)\in\mathcal{Y}=\left[0,1\right]$, then a mapping $f:{\cal X}\to{\cal S}$ is just a number in $\left[0,1\right]$. When the loss function is the absolute deviation
$
l\left(f,y\right):=\left|y-f\right|,
$
and the error function is mean absolute deviation
$
e_{\widetilde{P}}\left(f\right):=\E_{\widetilde{P}}\left[\left|Y-f\right|\right],
$
the true median $f^{*}$ minimizes the error, i.e. $
f^{*}\in\arg\min_{f\in\left[0,1\right]}e_{\widetilde{P}}\left(f\right).
$ However, it is not true that $
\left|f-f^{*}\right|=e_{\widetilde{P}}\left(f\right)-e_{\widetilde{P}}\left(f^{*}\right)
$ for any $f\in\left[0,1\right]$.
To see this, suppose that $Y\sim U\left[0,1\right]$ under $\widetilde{P}$. Then $f^{*}=0.5$
and 
$e_{\widetilde{P}}\left(f^{*}\right)=0.25$. However, for $f=0.4$, we have
$e_{\widetilde{P}}\left(f\right)=0.26.$
but
$
\left|f-f^{*}\right|=0.1\neq0.01=e_{\widetilde{P}}\left(f\right)-e_{\widetilde{P}}\left(f^{*}\right).
$

Moreover, there is no  function $d:\left[0,1\right]^{2}\to\left[0,1\right]$
such that decomposability \eqref{eq:decompose} holds, which would require that 
$
d\left(f,f_{\widetilde{P}}\right)=e_{\widetilde{P}}\left(f\right)-e_{\widetilde{P}}\left(f_{\widetilde{P}}\right)$ for any distribution $P$ of
$Y$ supported on $\left[0,1\right]$.
To see this, suppose that $Y\sim U\left[0,1\right]$ under $\widetilde{P}_1$,
we have
\[
e_{\widetilde{P}_{1}}\left(f\right)-e_{\widetilde{P}_{1}}\left(f_{\widetilde{P}_{1}}\right)=\left(f-0.5\right)^{2}=\left(f-f_{\widetilde{P}_{1}}\right)^{2},\quad \forall f \in [0,1].
\]
However, supposing that, under $\widetilde{P}_2$, the probability density function of $Y$ is given by $2y$ for $y\in\left[0,1\right]$,
we have
$
f_{\widetilde{P}_{2}}=\sqrt{2}/2$ and $e_{\widetilde{P}_{2}}\left(f_{\widetilde{P}_{2}}\right)=(2-\sqrt{2})/3$
but 
\[
e_{\widetilde{P}_{2}}\left(f\right)-e_{\widetilde{P}_{2}}\left(f_{\widetilde{P}_{2}}\right)=\frac{1}{3}\left(2f^{3}-3f^{2}+\sqrt{2}\right)\neq\left(f-f_{\widetilde{P}_{2}}\right)^{2}.
\]

\end{document}